\pdfoutput=1
\documentclass[letterpaper,11pt]{article}
\usepackage[utf8]{inputenc}

\usepackage[numbers]{natbib}
\usepackage{amsmath, amsthm, amssymb, thm-restate}
\usepackage{algorithmicx}
\usepackage[table,xcdraw]{xcolor}
\usepackage[ruled,vlined,linesnumbered]{algorithm2e}
\usepackage{mathtools}
\usepackage{comment} 
\usepackage{tcolorbox} 
\usepackage{xfrac}
\usepackage{multirow}
\usepackage{hyperref}
\usepackage{caption}
\usepackage{bm}
\usepackage{newfloat}
\usepackage{enumitem}

\usepackage{fancyhdr}

\usepackage[margin=1in]{geometry}

\allowdisplaybreaks

\definecolor{mygreen}{RGB}{20,150,80}
\definecolor{myred}{RGB}{150,0,0}
\definecolor{darkblue}{RGB}{0,0,150}

\hypersetup{
     colorlinks=true,
     citecolor= mygreen,
     linkcolor= myred,
     urlcolor=darkblue
}

\renewcommand{\epsilon}{\varepsilon}

\DeclareMathOperator{\E}{\ensuremath{\normalfont \textbf{E}}}

\newcommand{\hiddencomment}[1]{}

\newcommand{\true}[0]{\ensuremath{\textup{TRUE}}}
\newcommand{\false}[0]{\ensuremath{\textup{FALSE}}}
\newcommand{\VO}[0]{\ensuremath{\textup{VO}}}
\newcommand{\EO}[0]{\ensuremath{\textup{EO}}}

\newcommand{\RGMIS}[0]{\ensuremath{\textup{RGMIS}}}

\newcommand{\lowest}[0]{\ensuremath{\textup{LOWEST}}}

\newcommand{\cev}[1]{\reflectbox{\ensuremath{\vec{\reflectbox{\ensuremath{#1}}}}}}

\newcommand{\tspsize}[1]{\ensuremath{\tau}(#1)}
\newcommand{\pathcoversize}[1]{\ensuremath{\rho}(#1)}

\newcommand{\algone}[0]{\Cref{alg:path-cover11}}
\newcommand{\algtwo}[0]{\Cref{alg:path-cover22}}

\DeclareMathOperator{\poly}{poly}

\usepackage[noabbrev,nameinlink]{cleveref}
\crefname{result}{Result}{Results}
\crefname{lemma}{Lemma}{Lemmas}
\crefname{theorem}{Theorem}{Theorems}
\crefname{property}{Property}{Properties}
\crefname{claim}{Claim}{Claims}
\crefname{problem}{Problem}{Problems}
\crefname{definition}{Definition}{Definitions}
\crefname{observation}{Observation}{Observations}
\crefname{proposition}{Proposition}{Propositions}
\crefname{assumption}{Assumption}{Assumptions}
\crefname{line}{Line}{Lines}
\crefname{figure}{Figure}{Figures}
\creflabelformat{property}{(#1)#2#3}
\crefname{equation}{}{}
\crefname{section}{Section}{Sections}
\crefname{appendix}{Appendix}{Appendices}
\crefname{openproblem}{Open Problem}{Open Problems}
\crefname{algCounter}{Algorithm}{Algorithms}
\Crefname{algCounter}{Algorithm}{Algorithms}

\newtheorem{result}{Result}
\newtheorem{theorem}{Theorem}[section]
\newtheorem{lemma}[theorem]{Lemma}
\newtheorem{proposition}[theorem]{Proposition}

\newtheorem{problem}[theorem]{Problem}
\newtheorem{definition}[theorem]{Definition}
\newtheorem{claim}[theorem]{Claim}

\newtheorem{observation}[theorem]{Observation}

\definecolor{mylightgray}{RGB}{230,230,230}

\crefname{algCounter}{Algorithm}{Algorithms}
\Crefname{algCounter}{Algorithm}{Algorithms}

\algnewcommand{\IIf}[1]{\State\algorithmicif\ #1\ \algorithmicthen}
\algnewcommand{\EndIIf}{\unskip\ \algorithmicend\ \algorithmicif}

\newenvironment{graytbox}{
\par\addvspace{0.1cm}
\begin{tcolorbox}[width=\textwidth,
                  boxsep=5pt,
                  left=1pt,
                  right=1pt,
                  top=2pt,
                  bottom=2pt,
                  boxrule=0pt,
                  arc=0pt,
                  colback=mylightgray,
                  colframe=black,
                  ]
}{
\end{tcolorbox}
}

\newenvironment{whitetbox}{
\par\addvspace{0.1cm}
\begin{tcolorbox}[width=\textwidth,
                  boxsep=5pt,
                  left=1pt,
                  right=1pt,
                  top=2pt,
                  bottom=2pt,
                  boxrule=1pt,
                  arc=0pt,
                  colframe=black,
                  colback=white
                  ]
}{
\end{tcolorbox}
}

\newenvironment{myproof}{
\vspace{-0.5cm}
\begin{proof}
}{
\end{proof}
}

\newcounter{algCounter}

\makeatletter
\renewcommand{\paragraph}{%
  \@startsection{paragraph}{4}%
  {\z@}{10pt}{-1em}%
  {\normalfont\normalsize\bfseries}%
}
\makeatother

\makeatletter
\patchcmd{\@algocf@start}
  {-1.5em}
  {0pt}
  {}{}
\makeatother

\title{Sublinear Algorithms for TSP via Path Covers}

 \author{
 Soheil Behnezhad\\{\em Northeastern University} \and 
 Mohammad Roghani\\{\em Stanford University} \and
 Aviad Rubinstein\\{\em Stanford University} \and
 Amin Saberi\\{\em Stanford University}
 }

\date{}

\begin{document}

\maketitle

\begin{abstract}
We study sublinear time algorithms for the {\em traveling salesman problem} (TSP). First, we focus on the closely related {\em maximum path cover} problem, which asks for a collection of vertex disjoint paths that include the maximum number of edges. We show that for any fixed $\epsilon > 0$, there is an algorithm that $(1/2 - \epsilon)$-approximates the maximum path cover size of an $n$-vertex graph in $\widetilde{O}(n)$ time. This improves upon a $(3/8-\epsilon)$-approximate $\widetilde{O}(n \sqrt{n})$-time algorithm of Chen, Kannan, and Khanna [ICALP'20].
    
    \medskip
    Equipped with our path cover algorithm, we give an $\widetilde{O}(n)$ time algorithm that estimates the cost of $(1,2)$-TSP within a factor of $(1.5+\epsilon)$ which is an improvement over a folklore $(1.75 + \epsilon)$-approximate $\widetilde{O}(n)$-time algorithm, as well as a $(1.625+\epsilon)$-approximate $\widetilde{O}(n\sqrt{n})$-time algorithm of [CHK ICALP'20]. For graphic TSP, we present an $\widetilde{O}(n)$ algorithm that estimates the cost of graphic TSP within a factor of $1.83$ which is an improvement over a $1.92$-approximate $\widetilde{O}(n)$ time algorithm due to [CHK ICALP'20, Behnezhad FOCS'21]. We show that the approximation can be further improved to $1.66$ using $n^{2-\Omega(1)}$ time.
    
    \medskip    
    All of our $\widetilde{O}(n)$ time algorithms are information-theoretically time-optimal up to $\poly\log n$ factors. Additionally, we show that our approximation guarantees for path cover and $(1,2)$-TSP hit a natural barrier: We show better approximations require better sublinear time algorithms for the well-studied maximum matching problem.
\end{abstract}

\clearpage

\setcounter{tocdepth}{2}
\tableofcontents
\thispagestyle{empty}

\clearpage

\setcounter{page}{1}

\section{Introduction}

The {\em traveling salesman problem} (TSP) is a central problem in combinatorial optimization. Given a set $V$ of $n$ vertices and their pairwise distances, it asks for a Hamiltonian cycle of the minimum cost. In this paper, we study {\em sublinear time} algorithms for TSP. The algorithm is given query access to the distance pairs, and the goal is to estimate the solution cost in time sublinear in the input size (which is $\Theta(n^2)$).

TSP is NP-hard to approximate within a polynomial factor for an arbitrary distance function. As such, much of the work in the literature has been on more specific distance functions. Some notable examples include {\em graphic TSP} \cite{GharanSS11,MomkeS11,Mucha12,SeboV14,chen2020} where the distances are the shortest paths over an arbitrary unweighted undirected graph, {\em $(1, 2)$-TSP} \cite{AdamaszekMP18,chen2020,BermanK06,karp1972reducibility,MnichM18} where the distances are 1 or 2, and more generally {\em metric TSP} \cite{KarlinKG21,CzumajS-SIAM09,christofides1976worst,serdyukov1978nekotorykh} where the distances satsify triangle inequality. 

In 2003, \citet*{CzumajS-STOC04,CzumajS-SIAM09} showed that for any fixed $\epsilon > 0$, a $(1+\epsilon)$-approximation of the cost of metric minimum spanning tree (MST) and thus a  $(2+\epsilon)$-approximation of the cost of metric TSP can be found in $\widetilde{O}(n)$ time. Twenty years later, it still remains a major open problem to either break two-approximation in $n^{2-\Omega(1)}$ time or prove a lower bound.\footnote{See e.g. Open Problem 71 on \href{https://sublinear.info/index.php?title=Open_Problems:71}{sublinear.info} \cite{sublinear-info}.}
However, better bounds are known for both graphic TSP and $(1,2)$-TSP. In this paper, we present improved algorithms for these two well-studied variants of TSP. Our main tool to achieve this is an improved algorithm for the closely related {\em maximum path cover} problem which might be of independent interest.

\vspace{-0.2cm}
\paragraph{Maximum Path Cover:}  The maximum path cover in a graph is a collection of vertex disjoint paths with the maximum number of edges in it. The (almost) 1/2-approximate maximum matching size estimator of \citet*{behnezhad2021} immediately implies an (almost) 1/4-approximation for the maximum path cover problem in $\widetilde{O}(n)$ time.\footnote{The application of sublinear time maximum matching algorithms for approximating maximum path cover was first proposed by Gupta and Onak. See \cite{sublinear-info}.} This can be improved to an (almost) $(3/8 = .375)$-approximation using the {\em matching-pair} idea of \citet*{chen2020} in $\widetilde{O}(n \sqrt{n})$-time.\footnote{We note that even though a subsequent result of \citet*{behnezhad2021} improved the running time for maximal matchings and graphic TSP from $O(n\sqrt{n})$ in \cite{chen2020} to $\widetilde{O}(n)$, it is not immediately clear whether the same holds for path cover and $(1,2)$-TSP as they rely on a different notion of a matching pair.} Our first main contribution is an improvement over both of these results:

\begin{graytbox}
\begin{result}[Formally as \Cref{thm: disjoint-paths-algorithm}]\label{res:path-cover}
    For any $\epsilon > 0$, there is a randomized algorithm that w.h.p. $(1/2-\epsilon)$-approximates the size of maximum path cover in $\widetilde{O}(n \cdot \poly(1/\epsilon))$ time.
\end{result}
\end{graytbox}

Besides quantitavely improving prior work both in the running time and the approximation ratio, \cref{res:path-cover} reaches a qualitatively important milestone as well. First, the running time of \Cref{res:path-cover} is information-theoretically optimal up to $\poly\log n$ factors (the lower bound holds for any constant approximation --- see \cref{sec:hardness}). Second, its approximation ratio hits a rather important barrier. We give a non-trivial reduction that shows a $(1/2+\Omega(1))$-approximation in $\widetilde{O}(n)$ time for maximum path cover would imply the same bound for maximum matching in bipartite graphs. Such a result has remained elusive for matching, which is one of the most extensively studied problems in the literature of sublinear time algorithms. See \Cref{sec:lowerbound}.

It is also worth noting that in bounding the running time of our algorithm in \cref{res:path-cover}, we use connections to parallel algorithms. Such a connection was previously only used for matchings \cite{behnezhad2021}.

\setlength\extrarowheight{5pt}
\begin{center}
\begin{table}[ht]
\begin{tabular}{ | >{\centering\arraybackslash} m{7em} |  >{\centering\arraybackslash} m{10em}| >{\centering\arraybackslash} m{7em} | >{\centering\arraybackslash} m{14.3em} |}   
  \hline
  Running Time& Approximation Ratio & Metric & Reference\\ 
  \hline
   \hline
   $\widetilde{O}(n)$ & $1.75+\epsilon$ & (1,2) & Folklore\\
   \hline
   $\widetilde{O}(n\sqrt{n})$ & $1.625+\epsilon$ & (1,2) & \citet*{chen2020}\\
   \hline
   $\widetilde{O}(n)$ & $1.5+\epsilon$ & (1,2) & \textbf{This work} (\cref{res:12-TSP})\\
   \hline \hline
   $\widetilde{O}(n)$ & $1.929$ & Graphic & \citet*{chen2020}\\
   \hline
   $\widetilde{O}(n)$ & $1.834$ & Graphic & \textbf{This work} (\cref{res:graphicTSP})\\
   \hline
   $n^{2-\Omega(1)}$ & $1.667$ & Graphic & \textbf{This work} (\cref{res:graphicTSP2})\\
   \hline \hline
   $\Omega(n^2)$& $1 + \epsilon$ & (1,2) \& Graphic & \citet*{chen2020}\\
   \hline
   $n^{1+\Omega(1)}$ (Conditional) & $1.5 - \epsilon$  & (1,2) \& Graphic & \textbf{This work} (\Cref{cor:matching-via-12tsp})\\
   \hline
\end{tabular}
\captionsetup{justification=centering}
\caption{Comparison of running time and approximation ratio of our TSP algorithms and lower bounds with prior work.}
\label{tbl:comparision}
\end{table}
\end{center}

\vspace{-4em}
\paragraph{$(1,2)$-TSP:} The $(1,2)$-TSP problem has been studied extensively in the classical setting. In his landmark paper,  \citet*{karp1972reducibility} showed that $(1,2)$-TSP is NP-hard. \citet*{PapadimitriouY93} then proved its APX-hardness. Since then there has been a significant amount of work on $(1,2)$-TSP in the classical setting.  The current best known inapproximability bound for $(1, 2)$-TSP is 535/534 \cite{KarpinskiS12}. After a series of works, the best known polynomial time approximation is $8/7$ \cite{BermanK06} which can be implemented in $O(n^3)$ time \cite{AdamaszekMP18}. For sublinear time algorithms, an $\widetilde{O}(n)$-time (almost) $1.75$-approximation is folklore \cite{sublinear-info}. \citet*{chen2020} improved the approximation to (almost) $1.625$ in $\widetilde{O}(n\sqrt{n})$ time.

It is not hard to see that up to a small additive error of 1, $(1,2)$-TSP is equivalent to finding a maximum path cover on the weight-1 edges and then connecting their endpoints via weight-2 edges. A simple calculation shows that any $\alpha$-approximation for the maximum path cover problem leads to a $(2-\alpha)$-approximation for $(1,2)$-TSP. Our path cover algorithm of \Cref{res:path-cover} immediately implies the following result as a corollary:
\begin{graytbox}
\begin{result}\label{res:12-TSP}
    For any $\epsilon > 0$, there is a randomized algorithm that w.h.p. $(1.5+\epsilon)$-approximates the cost of $(1,2)$-TSP in $\widetilde{O}(n \cdot \poly(1/\epsilon))$ time.
\end{result}
\end{graytbox}

Similar to \Cref{res:path-cover}, the running time of \Cref{res:12-TSP} is information-theoretically optimal up to $\poly\log n$ factors, and its approximation ratio hits a natural barrier due to a connection to sublinear time matching that we establish in this work.

\paragraph{Graphic TSP:} The graphic TSP problem is equivalent to finding a tour of the minimium size that visits all the vertices. This is an important instance of TSP that has received a lot of attention over the years. For polynomial time algorithms, a $1.5$-approximation of \citet*{christofides1976worst} (which also works more generally for metric TSP) had remained the best known until a series of works over the last decade improved it to $(1.5 - \epsilon_0)$ \cite{GharanSS11}, $1.461$ \cite{MomkeS11}, $1.444$ \cite{Mucha12}, and finally to 1.4 \cite{SeboV14}. For sublinear time algorithms, \citet*{chen2020} showed that an (almost) $(27/14 \approx 1.928)$-approximation of graphic TSP can be obtained in $\widetilde{O}(n \sqrt{n})$ time. The running time was subsequently improved to $\widetilde{O}(n)$ by \citet*{behnezhad2021}.

We first show that plugging \cref{res:path-cover} into the framework of \cite{chen2020} immediately improves their approximation from $1.928$ to (almost) $1.9$ while keeping the running time $\widetilde{O}(n)$. We then give a more fine tuned algorithm that obtains a much improved approximation ratio of $(11/6 \approx 1.833)$.

\begin{graytbox}
\begin{result}\label{res:graphicTSP}
    For any $\epsilon > 0$, there is a randomized algorithm that w.h.p. $(1+\epsilon)(\frac{11}{6} \approx 1.833)$-approximates the cost of graphic TSP in $\widetilde{O}(n \cdot \poly(1/\epsilon))$ time.
\end{result}
\end{graytbox}

Over the past few years, significant advancements have been made in the development of sublinear matching algorithms. Several recent results \cite{BehnezhadRR23focs, behnezhadroghanirubinstein2022,brrstoc24, BehnezhadRRS-ArXiv22, bhattacharya2023dynamic, sayanSublinear2022} have led to the creation of a $(1, \epsilon n)$-approximation algorithm for maximum matching, with running time of $n^{2-\Omega_{\epsilon}(n)}$. Leveraging these sublinear algorithms, we have devised a slightly subquadratic algorithm that provides a more accurate estimation of the size of graphic TSP.

\begin{graytbox}
\begin{result}\label{res:graphicTSP2}
    For any $\epsilon > 0$, there is a randomized algorithm that w.h.p. $(1+\epsilon)(\frac{5}{3} \approx 1.666)$-approximates the cost of graphic TSP in $n^{2-\Omega_{\epsilon}(1)}$ time.
\end{result}
\end{graytbox}

We contrast our results with prior sublinear TSP algorithms in \Cref{tbl:comparision}.

\paragraph{Further related work:} Finally, we note that in a recent paper, \citet*{ChenMetric-Arxiv22} show that assuming that the metric has a spanning tree supported on weight 1 edges, one can obtain a $(2-\epsilon_0)$-approximation with $\widetilde{O}(n\sqrt{n})$ queries for some small unspecified constant $\epsilon_0 > 0$. While this is a more general metric than graphic TSP and (1,2)-TSP that we study in this paper, we note that the two papers are orthogonal and their results are incomparable. In particular, the techniques developed in this paper are specifically designed to improve the approximation to much below 2, whereas \cite{ChenMetric-Arxiv22} focuses on generalizing the distance function while beating 2.

\section{Technical Overview}\label{sec:techniques}

In this section, we give an overview of our algorithms, especially our sublinear time maximum path cover algorithm of \cref{res:path-cover} which is the key to the other results as well.

Let us start with using matchings to approximate maximum path cover. Consider a graph that has a Hamiltonian path. Here, the optimal maximum path cover has size $n-1$. On the other hand, any maximum matching can have at most $n/2$ edges, which is by a factor 2 smaller than our optimal path cover. On top of this, we only know close to 1/2 approximations for maximum matching if we restrict the running to be close to linear in $n$ \cite{behnezhad2021,BehnezhadRRS-ArXiv22}, thus can only achieve an approximation close to 1/4.

Instead of a single matching, \citet*{chen2020} showed how to estimate the number of edges in a {\em maximal matching pair} in $\widetilde{O}(n\sqrt{n})$ time, where a matching pair is simply two edge disjoint matchings. It is not hard to see that the number of edges in a maximal matching pair is at least half the number of edges in a maximum path cover. The problem, however, is that a maximal matching pair is not a collection of paths! In particular, the two matchings can form cycles of length as small as four. Therefore, one may only be able to use 3/4 fraction of the edges of a matching pair in a path cover. This is precisely why the algorithm of \cite{chen2020} only obtains a $\frac{1}{2} \times \frac{3}{4} = \frac{3}{8}$ approximation for path cover, and a $2-\frac{3}{8} = 1.625$ approximation for $(1,2)$-TSP.  

If we could modify the matching pair algorithm of \cite{chen2020}, and avoid cycles by manually excluding edges whose endpoints are the endpoints of a path in the current matching pair, then we could avoid the 3/4 factor loss discussed above and achieve a 1/2-approximation. Unfortunately, checking whether the endpoints of an edge are endpoints of a path requires knowledge about whether a series of other edges belong to the solution, which seems hard to implement in sublinear time. 

Instead of checking for cycles manually, we introduce the following \Cref{alg:path-cover11} which avoids cycles more naturally. While our final algorithm is a modified variant of \Cref{alg:path-cover11} described below, we start with \Cref{alg:path-cover11} as we believe it provides the right intuition.

\begin{algorithm}[H]
\caption{A new algorithm for path cover.}
\label{alg:path-cover11}

	Initialize $P \gets \emptyset$.
	
	Each vertex $v$ has two {\em ports} that we denote by $v^0$ and $v^1$. Each of these ports throughout the algorithm will be either {\em free} or {\em occupied}. Initially, all ports are free.
	
	Iterate over the edges in some ordering $\pi$. Upon visiting an edge $e=(u, v)$: 
	
	{
	\begin{itemize}[itemsep=-3pt,topsep=0pt,leftmargin=10pt]
	    \item If $v^0$ and $u^0$ are free, add $e$ to $P$, mark $v^0$ and $u^0$ as occupied, and skip to the next edge.
        \item If $v^1$ and $u^0$ are free, add $e$ to $P$, mark $v^1$ and $u^0$ as occupied, and skip to the next edge.
        \item If $v^0$ and $u^1$ are free, add $e$ to $P$, mark $v^0$ and $u^1$ as occupied, and skip to the next edge.
	\end{itemize}
	}
	
    Return $P$.
\end{algorithm}

Two properties of \Cref{alg:path-cover11} are crucial. First, it prioritizes occupying $(u^0, v^0)$ (compared to $(u^1, v^0)$ or $(u^0, v^1)$) which in particular implies that any component in $P$ must have a $(u^0, v^0)$ edge. Second, it never occupies $(u^1, v^1)$ with an edge $(u, v)$.  While it is easy to see that the output of \Cref{alg:path-cover11} has maximum degree 2, and is thus a collection of paths or cycles, the two properties above actually guarantee that it never includes any cycle. See \cref{fig:no-cycle}. We provide the formal proof of this later in \cref{sec:pathcover-algs}. Additionally, we show that the output of \Cref{alg:path-cover11} must be at least half the size of a maximum path cover, as we prove next. Hence, if we manage to estimate the size of the output $P$ of \Cref{alg:path-cover11}, then we have proved \Cref{res:path-cover}.

\begin{figure}[h]
    \centering
    \includegraphics[scale=0.85]{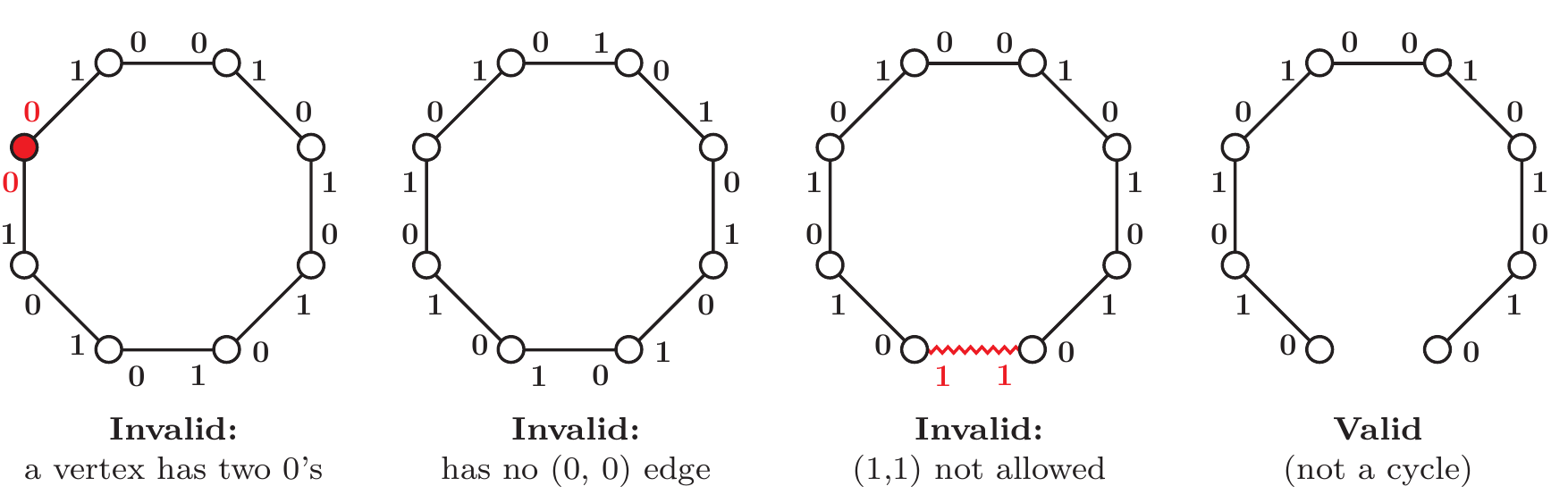}
    \caption{Examples of why the output of \Cref{alg:path-cover11} will not have cycles.}
    \label{fig:no-cycle}
\end{figure}

Our final algorithm is slightly different from \algone{} discussed above. In particular, we slightly relax it---see \algtwo{}---so that it can be solved via a randomized greedy maximal independent set (RGMIS), for which we have a rich toolkit of sublinear time estimators. Existing approaches (particularly the algorithm of \citet*{YoshidaYI09} and its two-step implementation by \citet*{chen2020}) can be employed to estimate the value of this modified \algtwo{} in $\widetilde{O}(n\sqrt{n})$ time. We achieve the improved, and near tight, $\widetilde{O}(n)$ time bound guarantee of \Cref{res:path-cover} by building on the analysis of \citet*{behnezhad2021} for maximal independent set on the line graphs (i.e., maximal matchings). Though we note that several new ideas are needed, because the MIS graph in our case will not be exactly a line graph. We defer more discussions about this to \cref{sec:pathcover-algs,sec:query-process}. 

\paragraph{Implications for TSP:} By having an $\alpha$-approximate maximum path cover algorithm, we immediately obtain a $(2-\alpha)$-approximation for $(1,2)$-TSP. Therefore, the algorithm above immediately proves \Cref{res:12-TSP} that we can (almost) $1.5$-approximate $(1,2)$-TSP in $\widetilde{O}(n)$ time. For our \Cref{res:graphicTSP} for graphic TSP, we first observe that our improved path cover algorithm can be employed to provide a better lower bound for the optimal TSP solution. This improves the 1.92-approximation of \cite{chen2020} as black-box to 1.9-approximation (\cref{sec:graphic-TSP-minor}). However, the final improvement to 1.83 requires more ideas, in particular, on how to better estimate the number of certain {\em bridges} in the graph. See \cref{sec:graphic-TSP-major} for more details about this.


\section{Preliminaries}

\paragraph{Problem Definition:} In the sublinear TSP problem, we have a set $V$ of $n$ vertices and a distance function $d: V \times V \to \mathbb{R}_+$. The algorithm has query access to this distance function. Namely, for any pair $(u, v)$ of the vertices of its choice, the algorithm may query the value of $d(u, v)$. The goal is to design an algorithm that runs in sublinear time in the input size, which is $\Theta(n^2)$ (all the distance pairs), and produces an estimate of the size of the optimal TSP solution. Denoting the optimal TSP value by $\tspsize{V}$, we say an estimate $\widetilde{\tau}(V)$ provides an $\alpha$-approximation for $\alpha \geq 1$ if
$$
    \tspsize{V} \leq \widetilde{\tau}(V) \leq \alpha \cdot \tspsize{V}.
$$
We focus specifically on {\em graph TSP} and {\em $(1,2)$-TSP} problems. In graphic TSP, the distance function $d$ is the shortest path metric on an unweighted undirected graph $G$ that is unknown to the algorithm. Note, however, that the distance queries essentially provide {\em adjacency matrix} access to this graph $G$. In $(1,2)$-TSP, the assumption is that $d(u, v) \in \{1,2\}$ for every pair $u, v$. In the case of $(1,2)$-TSP we may use $G$ to refer to the subgraph induced on the pairs with distance 1.

Defining graph $G$ as above, we use $n$ to denote the number of its vertices, $m$ to denote the number of its edges, $\Delta$ to denote its maximum degree, $\mu(G)$ to denote its maximum matching size, $\nu(G)$ to denote its minimum vertex cover size, and $\bar{d}$ to denote its average degree.

\paragraph{Path Cover Definitions:} Given an unweighted graph $G$, a path cover in $G$ is a collection of vertex disjoint paths in $G$. A maximum path cover is a path cover of $G$ with the maximum number of edges in it (note that we are not counting the number of paths, but rather the total number of edges in them). We use $\pathcoversize{G}$ to denote the size of the maximum path cover in $G$. We say an estimate $\widetilde{\rho}(G)$ for $\pathcoversize{G}$ provides an $(\alpha, \epsilon)$-approximation for $\alpha, \epsilon \in [0, 1]$ if
$$
    \alpha \cdot \pathcoversize{G} - \epsilon n \leq \widetilde{\rho}(G) \leq \pathcoversize{G}.
$$
We may also use $\alpha$-approximation instead of $(\alpha, 0)$-approximation.

\paragraph{Graph Theory Definitions/Tools:} A {\em bridge (cut edge)} in a graph is an edge whose deletion increases the number of connected components. Similarly, a {\em cut vertex} is a vertex whose deletion (along with its edges)  increases the number of connected components. A {\em biconnected graph} is a connected graph with no cut vertex. Also, a {\em biconnected component (block)} of a graph is a maximal biconnected subgraph of the original graph. A non-trivial biconnected component is a block that is not a bridge. We say a graph is {\em 2-edge-connected} if there is no bridge in the graph. A {\em 2-edge-connected component} of a graph is maximal 2-edge-connected subgraph of the original graph. The {\em bridge-block tree} of a graph is a tree obtained by contracting the 2-edge-connected components; note that the edge set of a bridge-block tree correspond to the bridges in the original graph.

We use the following classic theorem of K\"{o}nig \cite{Konig1916} that the size of the minimum vertex cover is equal to the size of maximum matching in bipartite graphs. Namely:

\begin{proposition}[K\"{o}nig’s Theorem]\label{prop:konig}
    In any bipartite graph $G$,  $\mu(G) = \nu(G)$.
\end{proposition}

\paragraph{Probabilistic Tools:} In our proofs, we use the following standard concentration inequalities.

\begin{proposition}[Chernoff Bound]\label{prop:chernoff}
    Let $X_1, X_2, \ldots, X_n$ be independent Bernoulli random variables. Let $X = \sum_{i=1}^{n} X_i$. For any $t > 0$, 
    $\Pr[|X - \E[X]| \geq t] \leq 2\exp\left(-\frac{t^2}{3\E[X]}\right).$
\end{proposition}

\begin{proposition}[Hoeffding’s Inequality]\label{prop:hoeffding}
    Let $X_1, X_2, \ldots, X_n$ be independent random variables such that $a \leq X_i \leq b$. Let $\bar{X} = (\sum_{i=1}^n X_i)/n$. For any $t > 0$, 
    $
    \Pr[|\bar{X} - \E[X]| \geq t] \leq 2\exp\left(-\frac{2nt}{(b-a)^2}\right).
    $
\end{proposition}

\section{New Meta Algorithms for Maximum Path Cover}\label{sec:pathcover-algs}

In this section, we present a new meta algorithm for maximum path cover that obtains a $1/2$-approximation. The algorithm, as we will state it in this section, will not be particularly in the sublinear time model. We discuss its sublinear time implementation later in \Cref{sec:query-process,sec: disjoint-paths-estimate}.

Our starting point is the \algone{} described in \cref{sec:techniques}. Let us first formally prove that it obtains a $1/2$-approximation, and that no component in it is a cycle.

\begin{claim}
    The output of \algone{} is a collection of disjoint paths.
\end{claim}
\begin{myproof}
    Since $P$ has maximum degree two, it suffices to show none of its connected components are cycles.  Property $(i)$  above implies that at any point during the algorithm, any degree one vertex $v$ has its port $v^0$ occupied. Now take an edge $e=(u, v)$ that forms a cycle if added to $P$. Both $u$ and $v$ must have degree one and so $u^0$ and $v^0$ are occupied. Since by property $(ii)$ edge $e$ does not occupy both $v^1$ and $u^1$, the algorithm does not add $e$ to $P$ thus not completing a cycle.
\end{myproof}

\begin{claim}\label{clm: alg-one-approx}
    Let $P^\star$ be {\em any} path cover using weight one edges. Then the output of \algone{} has size at least $\frac{1}{2}|P^\star|$.
\end{claim}
\begin{myproof}
    For any edge $e=(u, v) \in P^\star$ define $\phi(e) = \frac{1}{4}(\deg_P(u) + \deg_P(v))$. We first claim that for every edge $e=(u, v)$ in $G$, we have $\phi(e) \geq 1/2$ (or, equivalently, $\deg_P(u)+\deg_P(v) \geq 2$). This is clear for edges $e \in P$ due to the contribution of $e$ itself to its endpoints' degrees, so fix $e \not\in P$. Consider the time that we process $e=(u, v)$ in the algorithm and decide not to add it to $P$. We claim that out of $v^0, v^1, u^0, u^1$ at least two ports must be occupied. Suppose w.l.o.g. and for contradiction that only $v^x$ is occupied for $x \in \{0, 1\}$. Then $(u, v)$ can occupy $v^{1-x}$ and $u^x$ and be added to $P$. This contradicts $(u, v)$ not being added to $P$ and proves our claim that $\phi(e) \geq 1/2$. 
    
    From the discussion above, we get that
    $$
        \sum_{e \in P^\star} \phi(e) \geq \sum_{e \in P^\star} 1/2 = |P^\star|/2.
    $$
    
    Moreover, because every vertex has degree at most two in $P^\star$, we get 
    $$
    \sum_{e \in P^\star} \phi(e) = \frac{1}{4} \sum_{(u, v) \in P^\star} \deg_P(u)+\deg_P(v) \leq \frac{1}{4} \cdot 2 \sum_{v \in V} \deg_P(v) = |P|.
    $$
    
    The two inequalities above combined imply that $|P| \geq |P^\star|/2$.
\end{myproof}

As discussed, our final algorithm is different from \algone{} discussed above. One problem with \algone{} is that it cannot be cast as an instance of the randomized greedy maximal independent set (RGMIS) algorithm for which there is a rich toolkit of sublinear time estimators. To remedy this, we present a modified variant of \algone{} whose output is (almost) as good, but in addition can be modeled as an instance of RGMIS. We denote the output of RGMIS on a graph $G$ with a permutation $\pi$ on its vertices by $\RGMIS(G, \pi)$.

The algorithm is stated below as \algtwo{}. Similar to the output of \algone{}, the output of \algtwo{} can be verified to have maximum degree two. Thus, it is a collection of paths and cycles. But unlike \algone{}, the output of \algtwo{} can have cycles. This happens since, unlike \algone{}, each connected component of the output of \algtwo{} is not guaranteed to have an edge $(u, v)$ occupying both $u^0$ and $v^0$. Nonetheless, we are able to show that this bad event only happens for a small fraction of connected components of the output of \algtwo{} in expectation, and so once we remove one edge of each of these cycles, the resulting collection of disjoint paths has almost the same size.

\begin{algorithm}[H]
\caption{A modification of \algone{} that uses RGMIS.}
\label{alg:path-cover22}

\textbf{Parameter:} $K$ (think of it as a large constant integer).

Let $G=(V, E)$ be the subgraph of weight one edges. We construct a graph $H=(V_H, E_H)$ from $G$ on which we run RGMIS.
    
Each vertex in $H$ corresponds to an edge $e$ in $G$ and two {\em ports} (as in \algone{}) of the endpoints of $e$ that it occupies. Formally, for any $(u, v) \in E$ we have $K + 2$ vertices in $H$:
{
    \begin{itemize}[itemsep=-2pt,topsep=2pt]
        \item One vertex that corresponds to occuping $u^0$ and $v^1$.
        \item One vertex that corresponds to occuping $u^1$ and $v^0$.
        \item $K$ vertices that each corresponds to occuping $u^0$ and $v^0$.
    \end{itemize}
}

Consider two distinct vertices $a$ and $b$ in $H$ corresponding to edges $e_a$ and $e_b$ in $G$:
  
{ 
\begin{itemize}[itemsep=-2pt,topsep=2pt]
        \item If $e_a = e_b$ then we add an edge between $a$ and $b$ in $H$.
        \item If $e_a$ and $e_b$ share exactly one endpoint $v$ and both $a$ and $b$ occupy the same port of $v$, we add an edge between $a$ and $b$ in $H$.
    \end{itemize}
}
    
Find a randomized greedy maximal independent set $I$ of $H$.
    
Let $P$ be the set of edges in $G$ corresponding to the vertices in $I$.
    
Return $P$.

\end{algorithm}

\begin{observation}\label{obs:no-zero-cycle}
Let $C$ be a connected component in the output of \algtwo{}. If $C$ is a cycle, then every edge in $C$ occupies one 0-port and one 1-port (that is, no edge occupies two 0-ports).
\end{observation}
\begin{myproof}
    Suppose that $C$ has $n'$ vertices. Since each vertex in a cycle has degree two, both ports of each vertex in $C$ must be occupied. Hence, $n'$ 0-ports and $n'$ 1-ports of $C$ are occupied in total. Given that any edge occupies at least one 0-port by the algorithm, we cannot have an edge that occupies two 0-ports, or else we should occupy more 0-ports than 1-ports of $C$, which is a contradiction.
\end{myproof}

Next, we show that up to a factor of $(1+2/k)$ which is negligible for $K$ in the order $1/\epsilon$, the output of \algtwo{} is an (almost) 1/2-approximation of the maximum path cover value.

\begin{observation}\label{obs:one-zero-path}
Let $C$ be a connected component in the output of \algtwo{}. If $C$ is a path, then it contains at most one edge that occupies two 0-ports.
\end{observation}
\begin{myproof}
    Let $C$ be the path $(v_1, v_2, \ldots, v_r)$. Since the degree of any vertex $v_i$ for $1<i<r$ is two in the path, both ports of $v_i$ must be occupied. For $v_1$ and $v_r$, on the other hand, only one port is occupied. Hence, the total number of 0-ports that are occupied by $C$ minus the number of 1-ports occupied by it is at most two. This means that there is at most one edge that occupies two 0-ports since all other types of edges occupy exactly one 0-port and one 1-port.
\end{myproof}

\begin{lemma}\label{clm: approximation-bound}
let $P$ be the output of \algtwo{} on graph $G$. Then
$$
    \frac{1}{2} \rho(G) \leq \E|P| \leq \left(1+\frac{2}{K}\right) \rho(G),
$$
where the expectation is taken over the randomization of computing RGMIS in \algtwo{}.
\end{lemma}
\begin{proof}
    Let $P^*$ be a maximum path cover. For any edge $e=(u, v) \in P^\star$ define $\phi(e) = \frac{1}{4}(\deg_P(u) + \deg_P(v))$. With the exact same argument as in the proof of \Cref{clm: alg-one-approx}, we get that $\phi(e) \geq 1/2$, which implies
    \begin{align*}
        \sum_{e \in P^\star} \phi(e) \geq \sum_{e \in P^\star} 1/2 = \rho(G)/2.
    \end{align*}
    Since the degree of each vertex in $P$ is at most two, we get
    \begin{align*}
        \sum_{e \in P^\star} \phi(e) = \frac{1}{4} \sum_{(u, v) \in P^\star} \deg_P(u)+\deg_P(v) \leq \frac{1}{4} \cdot 2 \sum_{v \in V} \deg_P(v) = |P|.
    \end{align*}
    By combining above inequalities we get $ \frac{1}{2} \rho(G) \leq |P|$. Note that we do not need the randomization for the proof of the lower bound.
    
    By construction of $P$, every vertex has degree at most two in $P$. Hence, all connected components of $P$ are cycles and paths. We claim that at most $\frac{2}{K+2}$ fraction of connected components are cycles in expectation. Since the expected number of connected components is at most $\E|P|$, from this we get that the expected number of cycles is at most $2\E|P|/(K+2)$. By removing one edge from each cycle, we obtain a valid solution for maximum path cover problem. Thus,
    \begin{align*}
        \E|P| - \frac{2\E|P|}{K + 2} = \frac{K}{K+2} \E|P| \leq \rho(G) \qquad \Rightarrow \qquad \E|P| \leq \left(1+ \frac{2}{K}\right)\cdot\rho(G).
    \end{align*}
   
    So it remains to show that at most $\frac{2}{K+2}$ fraction of connected components are cycles in expectation. As we process edges one by one according to the ordering of RGMIS, let $A$ be the set of edges that none of their incident edges are added to the solution of \algtwo{}. By definition of $A$, if one copy of edge $(u, v)$ is in $A$, then all other copies of $(u,v)$ are also in $A$. Therefore, at any point during running RGMIS, if a new component is added to the solution, the edge $(u,v)$ that gets added to the solution occupies $(u^0, v^0)$ with probability at least $\frac{K}{K+2}$ since $K$ copies out of the $K+2$ copies are for $(u^0, v^0)$. Let $C_0$ be the number of times that the newly added component is an edge occupying two 0-ports, and $C_1$ be the number of times that the newly added component is an edge occupying one 0-port and one 1-port. By the above argument, we have
    \begin{align}\label{eq:cycle-ineq}
        \frac{\E[C_0]}{\E[C_0] + \E[C_1]} = \frac{K}{K+2}.
    \end{align}
    
    Note that after running \algtwo{}, it is possible that the number of connected components is actually smaller than $C_0 + C_1$, since some of the components may merge as the algorithm proceeds. However, by \Cref{obs:one-zero-path}, two components that their first edge occupies two 0-ports will not merge together. Also, by \Cref{obs:no-zero-cycle}, none of the cycle components have an edge that occupies two 0-ports. Therefore, in the end, there exists at most $\E[C_0] + \E[C_1]$ connected components and at least $\E[C_0]$ of them will not be cycles. This completes the proof.
    \end{proof}


\section{A Local Query Process for Algorithm 2 and its Complexity}\label{sec:query-process}

In this section, we define a query process to estimate the size of the output of \algtwo{}. 

In graph $H$ of \algtwo{}, each vertex corresponds to an edge in the original graph. More precisely, we make $K+2$ copies of each edge $(u, v)$ such that one of the copies corresponds to an edge occupying $(u^0, v^1)$, one for $(u^1, v^0)$, and $K$ for $(u^0, v^0)$. We use $G' = (V, E')$ to show the new graph with these parallel edges. During the course of \algtwo{}, two different edges that share the same endpoint and port cannot appear in the solution together. We use the following definition to formalize this notion.

\begin{definition}[Conflicting Pair of Edges]
Two edges $e, e' \in E'$ that share an endpoint $v$ are {\em conflicting} if both $e$ and $e'$ correspond to same port $v^i$ for $i \in \{0, 1\}$. We call $(e,e')$ a conflicting pair of edges.
\end{definition}

In order to estimate the size of the output of \algtwo{}, we define a vertex oracle that given a vertex $v$ and a permutation $\pi$ on $E'$, returns the degree of vertex $v$ in the output of \algtwo{}. These are akin to the query processes used before in the works of \cite{behnezhad2021,YoshidaYI09}, but are specific to our \algtwo{}.

\begin{algorithm}[H]
\caption{``vertex oracle'' $\VO(u, \pi)$ to determine the degree of vertex $u$ in $\RGMIS(G', \pi)$.}
\label{alg:vertexoracle}

	Let $e_1 = (u, v_1), \ldots, e_r = (u, v_r)$ be the edges incident to $u$ with $\pi(e_1) < \ldots < \pi(e_r)$.
	
	$d \leftarrow 0$
	
	\For{$i$ in $1 \ldots r$}{
		\lIf{$\EO(e_i, v_i, \pi) = \true$}{$d \leftarrow d+1$}
	}
	\Return $d$
\end{algorithm}

\begin{algorithm}[H]
\caption{``edge oracle'' $\EO(e, u, \pi)$ to determine an edge $e$ is in $\RGMIS(G', \pi)$. Also, $u$ must be an endpoint of $e$.}
\label{alg:edgeoracle}

    \lIf{$\EO(e, u, \pi)$ computed before}{\Return the computed result.}
    
	Let $e_1 = (u, v_1), \ldots, e_r = (u, v_r)$ be the edges incident to $e$ such that $\pi(e_1) < \ldots < \pi(e_r) < \pi(e)$. Also, $(e, e_i)$ is a conflicting pair for all $1 \leq i \leq r$.
	
	\For{$i$ in $1 \ldots r$}{
		\lIf{$\EO(e_i, v_i, \pi) =  \true$}{\Return \false}
	}
	
	\Return \true
\end{algorithm}

Note that in Line 2 of the \Cref{alg:edgeoracle} we only recursively call the function on edges that their label, conflict with edge $e$ since if other edges appear in the RMGIS subgraph, we can still have $e$ in the RGMIS subgraph. Before analyzing the query complexity of the vertex oracle, we prove the correctness of the vertex oracle.

\begin{claim}\label{clm:correctness-edge-oracle}
For any edge $e = (u, z) \in E'$ that is occupying ports $u^i$ and $z^j$, if $\EO(e, u, \pi)$ is called while computing $\VO(v, \pi)$, then $\EO(e, u, \pi) = \true$ iff $e \in \RGMIS(G', \pi)$.
\end{claim}
\begin{myproof}
We prove the claim using induction on ranking of edge $e$. Assume that the claim is true for all edges with ranking smaller than $\pi(e)$. If $\EO(e, u, \pi)$ is called by $\EO(e' = (w, z), z, \pi)$ or directly by $\VO(v, \pi)$, then by definition of \Cref{alg:edgeoracle} and \Cref{alg:vertexoracle}, all edges $e'' = (w', z)$ with $\pi(e'') < \pi(e')$ that are occupying $z^j$ are queried before $e'$ which means that none of them return \true{}. Hence, by induction hypothesis, none of the edges incident to $z$ that are occupying $z^j$ with lower rank are in the $\RGMIS(G', \pi)$. Moreover, $\EO(e, u, \pi)$ calls all incident edges to $u$ with lower rank that are occupying $u^i$ and return \true{} if none of them are in the $\RGMIS(G', \pi)$ by induction hypothesis. Therefore, $\EO(e, u, \pi) = \true$ iff $e \in \RGMIS(G', \pi)$.
\end{myproof}

\begin{claim}
Let $v \in V$ and $d$ be the output of $\VO(v, \pi)$. Then $d$ is equal to the degree of vertex $v$ in the subgraph outputted by $\RGMIS(G', \pi)$.
\end{claim}
\begin{myproof}
The observation follows by combining the fact that the vertex oracle queries edges in increasing order and \Cref{clm:correctness-edge-oracle}.
\end{myproof}

Let $T(v, \pi)$ denote the number of recursive calls to the edge oracle during the execution of $\VO(v, \pi)$.

\begin{theorem}\label{thm: query-complexity}
For a randomly chosen vertex $v$ and permutation $\pi$ on $E'$, we have that
\begin{align*}
    \E_{v, \pi}[T(v,\pi)] = O(\bar{d}\cdot \log^2 n)
\end{align*}
where $\bar{d}$ is the average degree of the graph $G$.
\end{theorem}

Let $Q(e, v, \pi)$ be the number of $\EO(e, \cdot, \pi)$ calls during the execution of $\VO(v, \pi)$. Moreover, let $Q(e, \pi)$ be the number of $\EO(e, \cdot, \pi)$ calls starting from any vertex. In other words, we have that $Q(e, \pi) = \sum_{v\in V}Q(e, v, \pi)$.

\begin{observation}\label{obs: q-upperbound}
For every edge $e$ and permutation $\pi$, $Q(e, \pi) \leq O(n^2)$.
\end{observation}
\begin{myproof}
Let $e = \{x,y\}$. For a fixed vertex $u$, either the vertex oracle $\VO(u, \pi)$ queries the edge oracle for $e$ directly, or through some incident edge $e'$. Hence, the edge oracle of $e$ is called through at most $(K + 2)(\deg(x) - 1) + (K + 2)(\deg(y) - 1)$ of its incident edges ($K+2$ appears since each edge has $K + 2$ copies), which implies that $Q(e, u, \pi) \leq (2K + 4)(n - 1) + 1$. Therefore,
\begin{align*}
    Q(e, \pi) \leq \sum_{u\in V} Q(e, u, \pi) \leq n\left( (2K + 4)(n - 1) + 1 \right) \leq O(n^2).\qedhere
\end{align*}
\end{myproof}

The main contribution of this section is to show that the expected number of $\EO(e, \pi)$ calls over all permutations $\pi$ is $O(\log^2 n)$, which is formalized in the following lemma.

\begin{lemma}\label{lem: query-complexity}
For any edge $e \in E'$, we have $\E_\pi[Q(e, \cdot, \pi)] = O(\log^2 n)$.
\end{lemma}

Assuming the correctness of \Cref{lem: query-complexity}, we can complete the proof of \Cref{thm: query-complexity}.

\begin{myproof}[Proof of \Cref{thm: query-complexity}]
\begin{align*}
    \E_{v, \pi}[T(v, \pi)] = \frac{1}{n}\E_\pi\Big[\sum_{v\in V}T(v,\pi)\Big] & = \frac{1}{n}\E_\pi\Big[\sum_{v\in V}\sum_{e \in E'}Q(e,v,\pi)\Big] \\
    & = \frac{1}{n}\E_\pi\Big[\sum_{e \in E'}\sum_{v\in V}Q(e,v,\pi)\Big] = \frac{1}{n} \E_\pi\Big[\sum_{e \in E'} Q(e, \pi)\Big] \\ 
    & = \frac{1}{n} \sum_{e \in E'} \E_\pi[ Q(e, \pi)] = \frac{1}{n} \sum_{e \in E'} O(\log^2 n) \\
    & = \frac{1}{n} O(|E'|\cdot \log^2 n) = O(\bar{d}\cdot \log^2 n).\qedhere
\end{align*}
\end{myproof}

During the recursive calls to the edge oracle that starts from vertex $v$, the edges in the stack of recursive calls create a trail. 

\begin{observation}
Let $S = (e_1 = (v, u), e_2, \ldots, e_r)$ be the stack of recursive calls starting from vertex $v$. Then $(e_1, e_2, \ldots, e_r)$ is a trail in $G'$.
\end{observation}
\begin{myproof}
Since in Line 2 of \Cref{alg:edgeoracle}, edge oracle only queries incident edges, $(e_1, e_2, \ldots, e_r)$ is a walk. It remains to show that all edges are distinct. Suppose that $e_i = e_j$ for some $i < j$ which implies $\pi(e_i) = \pi(e_j)$. Since the edge oracle queries edges in decreasing order, we have $\pi(e_j) < \pi(e_i)$ which is a contradiction. 
\end{myproof}

We direct the edges of the trail from $v$ to the other endpoint. We call a trail that starts from $v$ on the graph with edge permutation $\pi$, a $(v,\pi)$-query-trail. For an edge $e = (x,y)$, let $\vec{e}$ denote the directed edge from $x$ to $y$ and $\cev{e}$ denote a directed edge from $y$ to $x$. 

\begin{observation}\label{obs:increasing_trail}
Let $\vec{P} = (\vec{e_1}, \vec{e_2}, \ldots, \vec{e_k})$ be a $(v, \pi)$-query-trail; then $\pi(e_1) > \pi(e_2) > \ldots > \pi(e_k)$.
\end{observation}
\begin{myproof}
During the answering whether an edge is in $\RGMIS(G',\pi)$, \Cref{alg:edgeoracle} recursively calls on edges with 
$\pi$ values lower than the value of the current edge. Therefore, the stack of recursive calls will be decreasing with respect to $\pi$ values.
\end{myproof}

Let $Q(\vec{e}, \pi) \subseteq Q(e, \pi)$ be the set of all query trails that end at $\vec{e}$ (with the same direction). In what follows, we obtain a bound for the query complexity for $\vec{e}$. We use this lemma to prove \Cref{lem: query-complexity}.

\begin{lemma}\label{lem: query-complexity-directed}
For any edge $e$, we have $\E_\pi[Q(\vec{e}, \pi)] = O(\log^2 n)$.
\end{lemma}

\begin{myproof}[Proof of \Cref{lem: query-complexity}]
Since $Q(e, \pi) = Q(\vec{e}, \pi) \cup Q(\cev{e}, \pi)$, by \Cref{lem: query-complexity-directed} we have 
\begin{align*}
\E_\pi[Q(e, \pi)] \leq \E_\pi[Q(\vec{e}, \pi)] + \E_\pi[Q(\cev{e}, \pi)] = O(\log^2 n) + O(\log^2 n) = O(\log^2 n). \qedhere
\end{align*}
\end{myproof}

Given a permutation $\pi$ and a trail $\vec{P}=(\vec{e_1}, \vec{e_2}, \ldots, \vec{e_k})$, we define $\phi(\pi, \vec{P})$ to be another permutation $\sigma$ over the edges such that:
\begin{align*}
    (\sigma(e_1), \sigma(e_2), \ldots, \sigma(e_{k-1}), \sigma(e_k)) \coloneqq (\pi(e_2), \pi&(e_3), \ldots, \pi(e_{k}), \pi(e_{1}))\\
    \pi(e') = \sigma(e') \quad\quad \forall e' \notin \vec{P}
\end{align*}

Given an edge $\vec{e}$, by using the above mapping function we can construct a bipartite graph $H$ with two parts $A$ and $B$ such that each part has $|E'|!$ vertices showing different permutations of edges. For a permutation $\pi \in A$ and a $(v,\pi)$-query-trail $\vec{P}$ that ends at $\vec{e}$ for some arbitrary vertex $v$, we connect $\pi$ in $A$ to $\phi(\pi, \vec{P})$ in $B$. Note that by construction of $H$, $\text{deg}(\pi_A) = Q(\vec{e},\pi_A)$ for all $\pi_A \in A$, since we have a unique edge for each query-trail that ends at $\vec{e}$ with permutation $\pi_A$. Hence, in order to prove \Cref{lem: query-complexity}, it is sufficient to prove that $\E_{\pi_A \sim A}[\text{deg}_H(\pi_A)]=O(\log^2 n)$. Let $\mathcal{Q}(\vec{e}, \pi)$ be the set of all query-trails for permutation $\pi$ that ends at $\vec{e}$. Let $\beta = c \log^2 n$ for some large $c$. We partition permutations into two sets of {\em likely} and {\em unlikely} permutations called $L$ and $U$ as follows:
\begin{align*}
    L \coloneqq \left\{\pi \in \Pi \Big\lvert \max_{\vec{P} \in \mathcal{Q}(\vec{e}, \pi)} |\vec{P}| \leq \beta \right\} \quad\quad U \coloneqq \Pi \setminus L.
\end{align*}

Likely permutations are those permutations that the longest query-trail ending at $\vec{e}$ has length at most $\beta$ and unlikely permutations are the remaining permutations. Let $A_L$ be the set of vertices corresponding to the likely permutations in $A$ and $A_U$ be the set of vertices corresponding to the unlikely permutations. The intuition behind this partitioning is that the set of unlikely permutations makes up a tiny fraction of all permutations which is formalized in \Cref{lem: few-long-trail}.

\begin{lemma}\label{lem: few-long-trail}
If $c$ is a large enough constant, then we have $|A_U| \leq |E'|! / n^2$.
\end{lemma}

Before proving \Cref{lem: few-long-trail}, we introduce the parallel implementation of the greedy maximal independent set.

\noindent\textbf{Parallel Randomized Greedy Maximal Independent Set:} Let $G$ be a graph and $\pi$ be a permutation over its edges. In each iteration, we pick all vertices whose rank is less than all their neighbors and remove all their neighbors. We denote the number of rounds in this algorithm until $G$ becomes empty as \textit{round complexity} and we show it with $\rho(G, \pi)$.

It is clear that the output of the parallel randomized greedy MIS is the same as $\RGMIS(G, \pi)$. We have the following known result about the round complexity of parallel randomized greedy MIS.

\begin{lemma}[{\cite[Theorem~3.5]{Blelloch12}}]\label{lem: parallel-RMGIS-rounds}
For a uniformly random chosen permutation $\pi$ over edges of $G$, we have $\rho(G, \pi) = O(\log^2 n)$, with probability of at least $1 - \frac{1}{n^2}$.
\end{lemma}

In order to use the above lemma, we need to show that for an unlikely permutation, the round complexity is large and therefore, small fraction of permutations are unlikely as a result of \Cref{lem: parallel-RMGIS-rounds}.

\begin{claim}\label{clm: large-round-complexity}
Let $\vec{P}$ be query-trail in $G'$ with permutation $\pi$. Then $\rho(G', \pi) \geq \lfloor \frac{|\vec{P}|}{2}\rfloor$.
\end{claim}
\begin{myproof}
Let $\vec{P} = (\vec{e_1}, \vec{e_2}, \ldots, \vec{e_k})$ be a query-trail. By \Cref{obs:increasing_trail}, we have $\pi(e_1) > \pi(e_2) > \ldots > \pi(e_k)$, where $e_k$ is the last edge on the trail. Let $\rho(e)$ show the round in which edge $e$ is deleted by the parallel algorithm. If we can show that for $i < k - 1$, $\rho(e_i) > \rho(e_{i+2})$, then we have that $\rho(e_2) \geq \lfloor \frac{k}{2}\rfloor$ which completes the proof. We prove it using a contradiction. Assume that $\rho(e_i) \leq \rho(e_{i+2})$ for some $1 < i < k - 1$. Note that $\rho(e_{i+1}) \geq \rho(e_i)$, otherwise, when $e_{i+1}$ is deleted from the graph, one of its corresponding ports that is shared with $e_i$ and $e_{i+2}$ was occupied which implies that at least one of $e_i$ and $e_{i+2}$ should be deleted at the same time. Hence, in round $\rho(e_i)$, edge $e_{i+1}$ is still present in the graph. Therefore, $e_i$ is not a local minimum in round $\rho(e_i)$ and is deleted due to presence of an edge $e'$ in the solution. Note that $e' \neq e_{i+1}$ since $e_{i+1}$ is not the minimum edge because $e_{i+2}$ is still in the graph. If $e'$ is only incident to $e_i$, $\EO(e_{i-1}, \cdot, \pi)$ should call $\EO(e', \cdot, \pi)$ before $\EO(e_i, \cdot, \pi)$ since $e'$ is the local minimum in round $\rho(e_i)$ and therefore $\pi(e') < \pi(e_i)$. If $e'$ is incident to both $e_i$ and $e_{i+1}$, $\EO(e_i, \cdot, \pi)$ should call $\EO(e', \cdot,\pi)$ before $\EO(e_{i+1}, \cdot,\pi)$ since $e'$ is local minimum at round $\rho(e_i)$ and therefore $\pi(e') < \pi(e_{i+1})$. In both cases, the edge oracle terminates and will not query edge $e_{i+2}$. Hence, the assumption that $\rho(e_i) \leq \rho(e_{i+2})$ leads to a contradiction and the proof is complete.
\end{myproof}

Now we are ready to prove \Cref{lem: few-long-trail}.

\begin{proof}[Proof of \Cref{lem: few-long-trail}]
For each unlikely permutation $\pi \in U$, there exists a query-trail of length larger than $\beta$. By \Cref{clm: large-round-complexity}, we have $\rho(G, \pi) \geq \lfloor \frac{\beta + 1}{2}\rfloor$. Since $\beta = c\log^2 n$, by choosing $c$ large enough and \Cref{lem: parallel-RMGIS-rounds}, we have that $|U|/|\Pi| \leq 1/n^2$. Therefore, $|U| \leq |E'|!/n^2$ which implies that $|A_U| \leq |E'|!/n^2$ since $A_U$ represents vertices that correspond to unlikely permutations.
\end{proof}

Next, we show that each vertex $\pi_B \in B$, has at most $\beta$ neighbors between likely permutations in part $A$ in bipartite graph $H$.

\begin{lemma}\label{lem: few-likely-neighbors}
Let $\pi_Y$ be a vertex in $Y$. Then $\pi_Y$ has most $\beta$ neighbors in $X_L$. 
\end{lemma}

Before proving this lemma, we show how we can prove \Cref{lem: query-complexity-directed} using \Cref{lem: few-long-trail} and \Cref{lem: few-likely-neighbors}.

\begin{proof}[Proof of \Cref{lem: query-complexity-directed}]
Note that by \Cref{obs: q-upperbound}, degree of each vertex $\pi_A \in A$ is at most $O(n^2)$. Combining \Cref{lem: few-long-trail}, we have
$$E(A_U, B) \leq |E'|! / n^2 \cdot O(n^2) \leq O\big(|E'|!).$$

Moreover, by \Cref{lem: few-likely-neighbors}, each vertex $\pi_B \in B$ has at most $O(\beta)$ neighbors in $A_L$. Since $H$ is a bipartite graph,  $E(A_L, B) \leq O(\beta)\cdot |A_L|$. Therefore, sum of degrees of all vertices in $A$ is at most 
\begin{align*}
    E(A_L, B) + E(A_U, B) \leq O(\beta)\cdot|A_L| + O(|E'|!) \leq O(\beta\cdot|E'|!).
\end{align*}

For a random vertex in $A$, the expected degree is $\frac{O(\beta \cdot |E'|!)}{|E'|!} = O(|E'|)$. Combining with $\beta = c\log^2 n$ and $\deg(\pi_A) = Q(\vec{e},\pi_A)$ completes the proof.
\end{proof}

The rest of this section, we prove \Cref{lem: few-likely-neighbors}. Before proving \Cref{lem: few-likely-neighbors}, we prove that if two different query-trails that are mapped to two different permutations of $A_L$ to $\pi_B \in B$ by $\phi$, the shorter query-trail must be subgraph of the longer one.

\begin{lemma}\label{lem:no-branch}
Let $\pi$ and $\pi'$ be two different permutations, and $\vec{P}$ and $\vec{P}'$ be $(v,\pi)$- and $(v',\pi')$-query-trail, respectively, that both end at edge $\vec{e}$. If $\phi(\pi,\vec{P})=\phi(\pi',\vec{P}')$ and $|\vec{P}| \geq |\vec{P}'|$, then $\vec{P}'$ is a subgraph of $\vec{P}$.
\end{lemma}

We prove this lemma by series of observations and claims. Let $\vec{P} = (\vec{e_k}, \ldots, \vec{e_1})$ and $\vec{P}'=(\vec{e_r}', \ldots, \vec{e_1}')$ such that $e=e_1=e_1'$. If $\vec{P}'$ is not a subgraph of $\vec{P}$, then it must branch after an edge $\vec{e_b}'$. This means that $\vec{e_i}=\vec{e_i}'$ for $i \leq b$ and $\vec{e_{b+1}} \neq \vec{e_{b+1}}'$. Note that $\vec{e_{b+1}}$ and $\vec{e_{b+1}}'$ can be copy of the same edge.

\begin{observation}\label{obs:cons-conf}
Let $\pi$ be a random permutation over $E'$. For a $(u, \pi)$-query-trail, if $f$ and $f'$ are two consecutive edges in the trail, then $(f, f')$ is a conflicting pair.
\end{observation}
\begin{myproof}
Since the edge oracle calls $\EO(f', \cdot,\pi)$ in $\EO(f, \cdot, \pi)$, $(f, f')$ must be a conflicting pair.
\end{myproof}

\begin{observation}\label{obs:pair-conf}
Let $f_1, f_2, f_3$ be three different edges incident to some vertex $u$ and let $\pi$ be a random permutation over $E'$. Let $\vec{P_1}$ be a $(x, \pi)$-query-trail that calls $\EO(f_3, \cdot, \pi)$ in $\EO(f_1, \cdot, \pi)$. Also, let $\vec{P_2}$ be a $(y, \pi')$-query-trail that calls $\EO(f_3, \cdot, \pi')$ in $\EO(f_2, \cdot, \pi')$. Then $(f_1, f_2)$ is a conflicting pair.
\end{observation}
\begin{myproof}
By \Cref{obs:cons-conf}, $(f_1, f_3)$ is a conflicting pair. Assume that both $f_1$ and $f_3$ occupied port $u^i$. Moreover, since $(f_2, f_3)$ is a conflicting pair, then $f_2$ is also occupying $u^i$. Therefore, $(f_1, f_2)$ is a conflicting pair.
\end{myproof}

\begin{observation}\label{obs:equal_edge}
$\pi(e_b) = \pi'(e_{b+1})$.
\end{observation}
\begin{myproof}
Since $\vec{e_{b+1}}$ is not in $\vec{P}'$, we have that $\phi(\pi', \vec{P}')(e_{b+1})=\pi'(e_{b+1})$. Also, $\phi(\pi, \vec{P})(e_{b+1})=\pi(e_b)$ since $\phi(\pi, \vec{P})$ shifts edges of the trail $\vec{P}$ by one. Given that permutation $\phi(\pi, \vec{P})$ is equal to $\phi(\pi', \vec{P}')$, we have $\pi(e_b) = \pi'(e_{b+1})$.
\end{myproof}

Without loss of generality, we can assume that $\pi(e_b) \leq \pi'(e_b)$ since we did not make any difference between $\pi$ and $\pi'$ until this point.

\begin{observation}\label{obs:branch-begin}
$\pi'(e_{b+1}) < \pi'(e_b)$.
\end{observation}
\begin{myproof}
By combining \Cref{obs:equal_edge}, our assumption that $\pi(e_b) \leq \pi'(e_b)$, and the fact that $\pi'$ is a permutation, we have that $\pi'(e_{b+1}) < \pi'(e_b)$.
\end{myproof}

\begin{claim}\label{clm:equal-ranks}
If $\pi(f) < \pi(e_b)$ or $\pi'(f) < \pi(e_b)$ for some edge $\vec{f}$, then $\pi(f) = \pi'(f)$.
\end{claim}
\begin{myproof}
There are five different possible cases for $f$:
\begin{itemize}
    \item $\vec{f} \notin \vec{P} \cup \vec{P}'$: Since $\phi$ only changes the edge on the query-trail and $\phi(\pi, \vec{P})=\phi(\pi', \vec{P}')$, we have $\pi(f) = \pi'(f)$. 
    \item $\vec{f} \in \{\vec{e_1}, \ldots, \vec{e}_{b-1}\}$: Since $\phi(\pi, \vec{P})(e_{i+1}) = \phi(\pi', \vec{P}')(e_{i+1})$ for $1 \leq i < b$, we have $\pi(e_i) = \pi'(e_i)$. Hence, $\pi(f) = \pi'(f)$.
    \item $\vec{f}=\vec{e_b}$: In this case, condition $\pi(f) < \pi(e_b)$ does not hold since $\pi(f) = \pi(e_b)$. Also, $\pi'(f)=\pi'(e_b) \geq \pi(e_b)$. Therefore, condition $\pi'(f) < \pi(e_b)$ does not hold. 
    \item $\vec{f}\in \{\vec{e}_{b+1}, \ldots, \vec{e}_k\}$: By \Cref{obs:increasing_trail}, we have that $\pi(f) > \pi(e_b)$. Therefore, condition $\pi(f) < \pi(e_b)$ does not hold. Let $\vec{f}=\vec{e_i}$ for $i > b$. Since $\phi(\pi, \vec{P}) = \phi(\pi', \vec{P}')$, we have $\pi'(f) = \pi(e_{i-1}) \geq \pi(e_b)$. Therefore, none of the conditions in the claim statement holds.
    \item $\vec{f}\in \{\vec{e_{b+1}}', \ldots, \vec{e_r}'\}$: By \Cref{obs:increasing_trail}, we have that $\pi'(f) > \pi'(e_b)$. Combining by our assumption that $\pi'(e_b) \geq \pi(e_b)$, we have $\pi'(f) \geq \pi(e_b)$. Let $\vec{f}=\vec{e_i}'$ for $i > b$. Since $\phi(\pi, \vec{P}) = \phi(\pi', \vec{P}')$, we have that $\pi(f) = \pi'(e_{i-1}') \geq \pi'(e_b) \geq \pi(e_b)$. Therefore, none of the conditions in the claim statement holds.
\end{itemize}
The proof is thus complete.
\end{myproof}

\begin{claim}\label{clm:new-branch}
$e_{b+1} \in \RGMIS(G',\pi')$.
\end{claim}
\begin{myproof}
We prove the claim by contradiction. Assume that $e_{b+1} \notin \RGMIS(G',\pi')$. Hence, there exists an edge $f$ which is incident to $e_{b+1}$ such that $\pi'(f) < \pi'(e_{b+1})$. Thus, $\EO(e_{b+1}, \cdot, \pi')$ will recursively call $\EO(f, \cdot, \pi')$. Let $f$ be incident to $e_{i}$ and $e_{i+1}$ for $i \in \{b, b+1\}$. In the query-trail $\vec{P}$, $\EO(e_{i+1}, \cdot, \pi)$ calls $\EO(e_i, \cdot, \pi)$. Therefore, using the \Cref{obs:pair-conf}, we have that $(f,e_i)$ is a conflicting pair. Note that by \Cref{obs:equal_edge}, we have $\pi'(f) < \pi(e_b)$. Hence, $\pi(f) = \pi'(f) < \pi(e_b)$ by \Cref{clm:equal-ranks}. Since both permutations are identical for ranks lower than $\pi(e_b)$, edge $f$ must appear in $\RGMIS(G', \pi)$ and the query-trail $\vec{P}$ is not a valid query-trail since $\EO(e_i, \cdot, \pi)$ terminates upon calling $\EO(f, \cdot, \pi)$ (see \Cref{fig:branch}). 
\end{myproof}

\begin{figure}[htbp]
\begin{center}
  \includegraphics[scale=0.50]{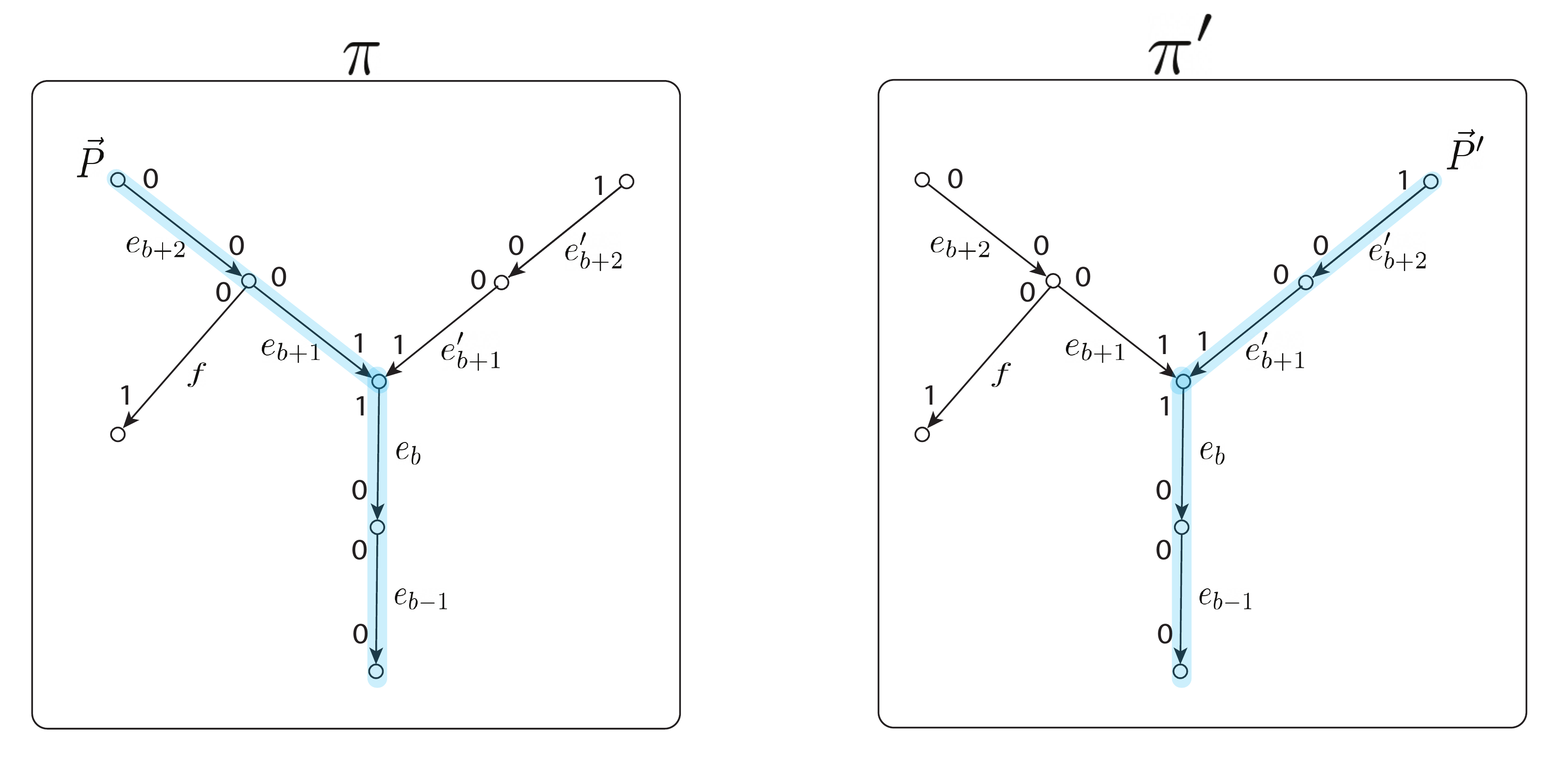}
  \caption{Illustration of proof of \Cref{clm:new-branch}. The highlighted blue trails show query-trails $\vec{P}$ and $\vec{P}'$. Query-trail $\vec{P}$ is not valid since $\EO(e_i, \cdot, \pi)$ terminates upon calling $\EO(f, \cdot, \pi)$.}\label{fig:branch}
  \end{center}
\end{figure}

\begin{myproof}[Proof of \Cref{lem:no-branch}]
 We prove that query-trail $\vec{P}'$ is not a valid $(v, \pi')$-query-trail. Note that by \Cref{obs:branch-begin}, $\EO(e_{b+1}', \cdot, \pi')$ calls $\EO(e_{b+1}, \cdot, \pi')$ before $EO(e_b, \cdot, \pi')$. Thus, by \Cref{clm:new-branch}, $\EO(e_{b+1}, \cdot, \pi')$ will return \true{} and execution of $\EO(e_{b+1}', \cdot, \pi')$ terminates at this point. Therefore, query-trail $\vec{P}'$ is a subgraph of query-trail $\vec{P}$.
\end{myproof}

Now we are ready to complete the proof of \Cref{lem: few-likely-neighbors}.

\begin{proof}[Proof of \Cref{lem: few-likely-neighbors}]
For each edge between $\pi_A \in A_L$ and $\pi_B \in B$ in graph $H$, we write a label $\chi(\pi_A, \pi_B)$ on the edge which is equal to the length of the query-trail corresponding to this edge in $H$. By \Cref{lem:no-branch}, all the labels for edges of a fixed vertex $\pi_B \in B$ that are incident to $A_L$ should be different. Moreover, by the definition of likely permutations, all query-trails of permutation $A_L$ have length less than or equal to $\beta$. Thus, each vertex $\pi_B \in B$ has at most $\beta$ neighbors in $A_L$.
\end{proof}

\section{Our Estimator for Maximum Path Cover}\label{sec: disjoint-paths-estimate}
In this section, we use the oracle of the previous section to estimate the number of edges in the output of \algtwo{}. In \Cref{sec:query-process}, we provide a lower bound on the number of recursive calls to our local query process. Note that this bound does not necessarily imply the same running time algorithm. For example, if we generate the whole permutation over all copies of edges before running the algorithm, it takes $\Theta(m)$ which is no longer sublinear. Using by now standard ideas of the literature, we show in \Cref{sec:implementation} how we can implement the query process in almost the same running time (multiplied by a polylogarithmic factor) which is formalized in the following lemma.

\begin{restatable}{lemma}{oracleImplementation}\label{lem: implementation}
There exists a data structure that given a graph $G$ in the adjacency list format, (implicitly) fixes a random permutation $\pi$ over its edges. Then for any vertex $v$, the data structure returns the degree of vertex $v$ in the subgraph $P$ produced by \algtwo{} according to a random permutation $\pi$. Each query $v$ to the data structure is answered in $\Tilde{O}(T(v, \pi))$ time w.h.p. where $T(v, \pi)$ is as defined in \cref{sec:query-process}.
\end{restatable}

Note that in our local query process, we need access to the adjacency list of weight-one edges. So the challenge that arises here is how to estimate the size of the output of \algtwo{} in the adjacency matrix model. We present a reduction from adjacency matrix to adjacency list that appeared in the literature \cite{behnezhad2021}. In this reduction, each query to the adjacency list can be implemented with $O(1)$ queries to the adjacency matrix and still we are able to estimate the maximum path cover with some additive error.

Let $\gamma = 16Kn$. We construct a graph $\hat{G} = (V_{\hat{G}}, E_{\hat{G}})$ for weight-one edges of graph $G$ as follows:

\begin{itemize}
    \item $V_{\hat{G}}$ is the union of $V_1, V_2$ and $U_1, U_2, \ldots, U_n$ such that:
    \begin{itemize}
        \item $V_1$ and $V_2$ are two copies of the vertex set of the original graph $G$.
        \item $U_i$ is a vertex set of size $\gamma$ for each $i \in [n]$.
    \end{itemize}
    \item We define the edge set such that degree of each vertex is in $\{1, n, n + \gamma \}$:
    \begin{itemize}
        \item Degree of each vertex $v \in V_1$ is $n$. The $i$-th neighbor of $v$ is the $i$-th vertex in $V_1$ if $(v, i) \in E$, otherwise its $i$-th neighbor is the $i$-th vertex in $V_2$ for $i \leq n$. Note that graph $(V_1, E_H \cap (V_1 \times V_1))$ is isomorphic to $G$.
        \item Degree of each vertex $v \in V_2$ is $n + \gamma$. The $i$-th neighbor of $v$ is the $i$-th vertex in $V_2$ if $(v, i) \in E$, otherwise, its $i$-th neighbor is the $i$-th vertex in $V_1$ for $i \leq n$. For all $n < i \leq n + \gamma$, the $i$-th neighbor of $v$ is $i$-th vertex in $U_v$.
        \item Degree of each vertex $u \in U_i$ is one for $i \in [n]$. The only neighbor of $u$ is the $i$-th vertex of $V_2$.
    \end{itemize}
\end{itemize}

By the construction of $\hat{G}$, the only neighbor of $v \in \bigcup_{i=1}^n U_i$ can be determined without any query to the adjacency matrix. Also, the $i$-th neighbor of each vertex in $V_1 \cup V_2$ can be determined with one query.

\begin{observation}\label{obs: finding-neighbor}
For each vertex $v \in V_{\hat{G}}$ and $i \in [\deg_{\hat{G}}(v)]$, the $i$-th neighbor of vertex $v$ can be determined using at most one query to the adjacency matrix.
\end{observation}

Fix a vertex $v \in V_2$. When we run \algtwo{}, intuitively with high probability the first edge that is incident to $v$ and occupies port $v^0$ is between $v$ and $u \in U_v$. Furthermore, with high probability the first two edges that are incident to $v$ and occupies port $v^1$ are between $v$ and $u \in U_v$. A vertex $v \in V_2$ is an {\em abnormal} vertex if the above properties do not hold for $v$. Let $R \in V_2$ be the set of abnormal vertices. In the following observation, we show that for each vertex $v \in V_2 \setminus R$, all incident edges of $v$ in the output of \algtwo{} are between $v$ and vertices of $U_v$.

\begin{claim}\label{clm:abnormal-vertices-count}
$\E_\pi|R| \leq n/(4K)$
\end{claim}
\begin{myproof}
Fix a vertex $v \in V_2$. For a random permutation over copies of edges of $\hat{G}$, the first incident edge to $v$ that occupies port $v^0$ is between $v$ and $U_v$ with a probability of at least $\frac{(K + 1)\gamma}{(n + \gamma)(K+1)} \geq 1 - \frac{1}{8K}$. Moreover, the first two edges that occupy $v^1$ are between $v$ and $U_v$ with probability of at least $\frac{\gamma(\gamma - 1)}{(n+\gamma)(n+\gamma - 1)} \geq 1 - \frac{1}{8K}$. Since both events are independent, the probability of $v$ not being an abnormal vertex is at least
$$
\left(1 - \frac{1}{8K}\right)^2 \geq 1 - \frac{1}{4K}
,$$
which implies $\E_\pi|R| \leq n/(4K)$.
\end{myproof}

\begin{claim}\label{obs: good-vertices}
For each $v \in V_2 \setminus R$, all incident edges of $v$ in the output of \algtwo{} are between $v$ and vertices of $U_v$.
\end{claim}

\begin{myproof}
By definition of an abnormal vertex, let the first edge in the permutation incident to $v$ be between $v$ and $w \in U_v$ which occupies $v^0$. Since all copies of edges incident to $w$ are between $v$ and $w$, this edge will be added to the solution of \algtwo{}. Moreover, we know that the first two edges that are incident to $v$ and occupy port $v^1$ are between $v$ and $U_v$. Let these two edges be $(v, u_1)$ and $(v, u_2)$ where $u_1, u_2 \in U_v$. Note that the only way that $(v, u_1)$ is not added to the solution of \algtwo{} is when $u_1 = w$. In this case, since there is only one copy for each edge that occupied port $v^1$, then $u_2 \neq w$. Therefore, \algtwo{} adds $(v, u_2)$ to its output if it has not added $(v, u_1)$. Since both ports of $v$ are occupied in this case, all incident edges of $v$ in the output of \algtwo{} are between $v$ and vertices of $U_v$. 
\end{myproof}

\begin{observation}\label{obs: small-TSP-R}
Let $P$ be the output of \algtwo{} on $\hat{G}$. Then
$$
    \frac{1}{2} \rho(\hat{G}[V_1 \cup R]) \leq \E \left|P \cap (V_1 \cup R) \times (V_1 \cup R) \right| \leq (1+\frac{2}{K})\cdot \rho(\hat{G}[V_1 \cup R])
.$$
\end{observation}
\begin{myproof}
By \Cref{obs: good-vertices}, if we run \algtwo{} on $\hat{G}$, for any vertex $v \in V_2 \setminus R$, all incident edges of $v$ in the output are between $v$ and $U_v$. Hence, none of the edges between $V_2 \setminus R$ and $V_1 \cup R$ will be added to the output of \algtwo{}. Since, the permutation over edges of $V_1 \cup R$ is uniformly at random, by \Cref{clm: approximation-bound}, we obtain the claimed bound.
\end{myproof}

In the above sequence of observations, we show that there are few abnormal vertices in $V_2$, which implies that most of the incident edges to vertices of $V_1$ in the output of \algtwo{} are in $\hat{G}[V_1]$ (only those between $V_1$ and $R$ violate this property). Therefore, a natural way to estimate the number of edges in the output of \algtwo{} on $G$, is to estimate the number of edges in $\hat{G}[V_1]$ in the output of \algtwo{} on $\hat{G}$. With this intuition in mind, we need to bound the query complexity of the algorithm for a random vertex in $V_1$.

\begin{claim}\label{clm: v1query}
Let $v$ be a random vertex in $V_1$ and $\pi$ be a random permutation over edges of graph that is created by copying $E_{\hat{G}}$ according to \algtwo{}. Then
\begin{align*}
    \E_{v\sim V_1, \pi}[T(v, \pi)] = \tilde{O}(n).
\end{align*}
\end{claim}
\begin{myproof}
By \Cref{thm: query-complexity}, we have that
\begin{align*}
\E_{v\sim V_{\hat{G}}, \pi}[T(v, \pi)] = O(\frac{K\cdot |E_{\hat{G}}|}{|V_{\hat{G}}|}\cdot \log^2 |V_{\hat{G}}|).
\end{align*}
Summing over all vertices in $V_{\hat{G}}$, we obtain
\begin{align*}
    \sum_{v \in V_{\hat{G}}}\E_{\pi}[T(v, \pi)] = O(K \cdot |E_{\hat{G}}| \cdot \log^2 |V_{\hat{G}}|) = \tilde{O}(n^2),
\end{align*}
since $|V_{\hat{G}}| = O(n^2)$, $K = O(1/\epsilon)$, and $|E_{\hat{G}}| = O(n^2)$. Therefore, for a random vertex in $V_1$, we get
\begin{align*}
    \E_{v\sim V_1,\pi}[T(v, \pi)] \leq \left( \sum_{v \in V_{\hat{G}}}\E_{\pi}[T(v, \pi)] \right) / |V_1| = \tilde{O}(n). \qedhere
\end{align*}
\end{myproof}

\begin{algorithm}
\caption{Final algorithm for maximum path cover.}
\label{alg: disjoint-paths}
    Let $\hat{G} = (V_{\hat{G}}, E_{\hat{G}})$ as described above.

    $r \gets 192\cdot K^2 \cdot \log n$.
    
    Sample $r$ vertices $u_1, u_2, \ldots, u_r$ uniformly at random from $V_1$ with replacement.
    
    Sample $r$ ports $p_1, p_2, \ldots, p_r$ uniformly at random from $\{0, 1\}$.
    
    Run vertex oracle for each $u_i$ and let $X_i$ be the indicator if port $u_i^{p_i}$ is occupied with an edge in $\hat{G}[V_1]$ in output of \algtwo{}.

Let $X = \sum_{i \in [r]} X_i$ and $f = X/r$.

Let $\tilde{\rho} = \frac{K}{2(K + 2)}\cdot (f \cdot n - \frac{n}{4K})$.
    
\Return $\tilde{\rho}$
\end{algorithm}

\begin{lemma}\label{lem: error-bound}
Let $\tilde{\rho}$ be the output of \Cref{alg: disjoint-paths} on input graph $G$. With high probability, 
$$ \left(\frac{1}{2} - \frac{1}{K}\right)\cdot \rho(G) - \frac{n}{K} \leq \tilde{\rho} \leq \rho(G),$$
where $K$ is the parameter which is defined in \algtwo{}.
\end{lemma}
\begin{myproof}
Let $\hat{P}$ be the set of edges outputted by \algtwo{} on $\hat{G}$ with both endpoints in $V_1$. By \Cref{clm: approximation-bound}, we have that $\E|\hat{P}| \leq (1 + \frac{2}{K})\cdot\rho(G)$. Furthermore, by \Cref{clm:abnormal-vertices-count} and the fact that the degree of each vertex in the output of \algtwo{} is at most two, in the output of \algtwo{} on $\hat{G}[V_1 \cup R]$ we have at most $n/(2K)$ edges with one endpoint in $R$. Hence, combining with \Cref{clm: approximation-bound} and \Cref{obs: small-TSP-R} we get
\begin{align}\label{eq: phat-range}
    \frac{1}{2}\rho(G) - \frac{n}{2K} \leq \E|\hat{P}| \leq (1 + \frac{2}{K})\cdot \rho(G).
\end{align}
Since each edge in the output of \algtwo{} counted twice in \Cref{alg: disjoint-paths}, we have 
\begin{align*}
    \E[X_i] = \Pr[X_i = 1] = \frac{2 \E|\hat{P}|}{n},
\end{align*}
and,
\begin{align}\label{eq: X-expected2}
    \E[X] = \frac{2r \E|\hat{P}|}{n}.
\end{align}
Since $X$ is sum of $r$ independent random variables, by Chernoff bound (\Cref{prop:chernoff}) we get
\begin{align*}
    \Pr[|X - \E[X]| \leq \sqrt{6 \E[X] \log n}] \leq 2 \exp \left(-\frac{6 \E[X] \log n}{3 \E[X]} \right) = \frac{2}{n^2}.
\end{align*}
Combining $fn = Xn/r$ and the above bound, with probability of at least $1 - 2/n^2$ we have
\begin{align*}
    fn \in \frac{n(\E[X] \pm \sqrt{6 \E[X] \log n)}}{r} & = \frac{n\E[X]}{r} \pm \sqrt{\frac{6 n^2\E[X] \log n}{r^2}}\\
    & = 2\E|\hat{P}| \pm \sqrt{\frac{12n\E|\hat{P}| \log n}{r}} & (\text{By }(\ref{eq: X-expected2}))\\
    & = 2\E|\hat{P}| \pm \sqrt{\frac{n\E|\hat{P}|}{16K^2}} & (\text{Since } r = 192 \cdot K^2 \cdot \log n) \\
    & \in 2\E|\hat{P}| \pm \frac{n}{4K} & (\text{Since } \E|\hat{P}| \leq n).
\end{align*}
Since, $\tilde{\rho} = \frac{K}{2(K+2)}\cdot (f \cdot n - \frac{n}{4K})$, hence
\begin{align*}
    \frac{K}{K+2}\left(\E|\hat{P}| - \frac{n}{2K} \right) \leq \tilde{\rho} \leq \frac{K}{K+2} \cdot \E| \hat{P}|.
\end{align*}
Combining with (\ref{eq: phat-range}), implies the claimed bound.
\end{myproof}

\begin{theorem}\label{thm: disjoint-paths-algorithm}
Given an adjacency matrix access for input graph $G$, there exists a randomized algorithm that w.h.p. runs in $\widetilde{O}(n)$ time and  produces an estimate $\tilde{\rho}$, such that
$$
\left(\frac{1}{2} - \epsilon\right)\cdot\rho(G) - \epsilon n \leq \tilde{\rho} \leq \rho(G).
$$
\end{theorem}
\begin{myproof}
Let $K = \frac{1}{\epsilon}$ and $\tilde{\rho}$ be the output of \Cref{alg: disjoint-paths} on $G$. In \Cref{alg: disjoint-paths}, by combining \Cref{lem: implementation} and \Cref{clm: v1query}, the running time for each sample is $\tilde{O}(n)$. Since the number of samples is $r = 192K^2\log n$, and $K$ is a constant, the total running time of the algorithm is $\tilde{O}(n)$. Moreover, by \Cref{lem: error-bound} we get the approximation ratio in the statement.
\end{myproof}

\section{Our Estimator for (1,2)-TSP}
In this section, we use the algorithm for estimating the size of maximum path cover as a black box to estimate the size of (1,2)-TSP. First, note that if there is no Hamiltonian cycle with weight one edges of the graph, then the set of weight-one edges of the graph (1,2)-TSP is a solution for maximum path cover of graph $G$. Also, in the case that there exists a Hamiltonian cycle, then the size of maximum path cover is $n-1$. Moreover, if $P^*$ is the maximum path cover of a graph $G$, then it is possible to create a TSP by connecting these paths using edges with weight two. This intuition helps to formalize the following observation.

\begin{observation}\label{obs: bounding-TSP-with-DP}
Let $\tspsize{V}$ be the cost of (1,2)-TSP of graph $G = (V, E)$. Then, we have 
$$2n - \rho(G) - 1 \leq \tspsize{V} \leq 2n - \rho(G).$$
\end{observation}

Now we are ready to present the final algorithm for estimating (1,2)-TSP.

\begin{algorithm}
\caption{Final algorithm for (1,2)-TSP.}
\label{alg: 12tsp}

    Construct $\hat{G} = (V_{\hat{G}}, E_{\hat{G}})$ implicitly as desribed in \Cref{sec: disjoint-paths-estimate}.
    
    Let $\tilde{\rho}$ be the output of \Cref{alg: disjoint-paths} on $\hat{G}$.
    
    $\tilde{\tau} = 2n - \tilde{\rho}$
    
    \Return $\tilde{\tau}$
\end{algorithm}

\begin{lemma}\label{lem: 12tsp-approx}
Let $\tilde{\tau}$ be the output of \Cref{alg: 12tsp} and $\tspsize{V}$ be the cost of (1,2)-TSP of graph $G = (V, E)$. With high probability,
$$
\tspsize{V} \leq \tilde{\tau} \leq (\frac{3}{2} + \frac{4}{K})\cdot \tspsize{V},
$$
where $K$ is the parameter which is defined in \algtwo{}.
\end{lemma}
\begin{myproof}
By \Cref{obs: bounding-TSP-with-DP}, we have $2n - \rho(G) - 1 \leq \tspsize{V} \leq 2n - \rho(G)$. \Cref{alg: 12tsp} outputs $\tilde{\tau} = 2n - \tilde{\rho}$ as the estimate, where $\tilde{\rho}$ is the output of \Cref{alg: disjoint-paths}. Hence, by \Cref{lem: error-bound}, we have $2n - \tilde{\rho} \geq 2n - \rho(G)$. Also, by \Cref{lem: error-bound}, we have 
\begin{align*}
    2n - \tilde{\rho} & \leq 2n - (\frac{1}{2} - \frac{1}{K})\cdot \rho(G) + \frac{n}{K} \\
    & \leq 3n - \frac{3\rho(G)}{2} + \frac{4n}{K} - \frac{2\rho(G)}{K} - 1 & (\text{Since } \rho(G) < n)\\
    & \leq (\frac{3}{2} + \frac{4}{K})(2n - \rho(G) - 1) & (\text{Since } K \ll n) \\
    & \leq (\frac{3}{2} + \frac{4}{K})\cdot\tspsize{V} & (\text{Since } \tspsize{V} = 2n - \rho(G)),\\
\end{align*}
which completes the proof.
\end{myproof}

\begin{theorem}\label{thm: 12tsp}
Let $\tspsize{V}$ be the cost of (1,2)-TSP of graph $G = (V, E)$. For any $\epsilon > 0$, there exists an algorithm that estimate the cost of (1,2)-TSP, $\tilde{\tau}$, such that
$$
\tspsize{V} \leq \tilde{\tau} \leq (\frac{3}{2} + \epsilon)\cdot \tspsize{V},
$$
w.h.p in $\tilde{O}(n)$ running time.
\end{theorem}
\begin{myproof}
We choose $K = \frac{4}{\epsilon}$. By \Cref{lem: 12tsp-approx}, if $\tilde{\tau}$ is the output of \Cref{alg: 12tsp}, we get
\begin{align*}
    \tspsize{V} \leq \tilde{\tau} \leq (\frac{3}{2} + \epsilon)\cdot \tspsize{V}.
\end{align*}
Also, since the running time of \Cref{alg: 12tsp} is the same as the running time of \Cref{alg: disjoint-paths}, by \Cref{thm: disjoint-paths-algorithm}, the total running time is $\tilde{O}(n)$, which completes the proof.
\end{myproof}

\section{Our Estimator for Graphic TSP}\label{sec:graphic-TSP-minor}
In this section, we use our algorithm for estimating the size of maximum path cover to estimate the size of graphic TSP. In a recent work, Chen et al. \cite{chen2020} showed that it is possible to obtain a $(27/14)$-approximate algorithm for graphic TSP by estimating the matching size and the number of biconnected components in the graph. Since the size of graphic TSP is at most $2n$ (the cost of MST is $n-1$), they proved that if a graph has large matching and a few biconnected components, the cost of graphic TSP is significantly lower than $2n$. Since estimating the number of biconnected components is not an easy task in sublinear time, they use a proxy quantity that can be estimated in sublinear time.

We show that if we use our estimator for maximum path cover as a black-box instead of matching estimator in algorithm of \cite{chen2020}, we can improve the approximation ratio to $19/10$. Moreover, we show that we can estimate the number of bridges in $\tilde{O}(n)$. We exploit this estimation for further improvement to get a $11/6$-approximate algorithm for graphic TSP. 

Chen et al. \cite{chen2020} introduced the following definition of {\em bad vertex} as a proxy for estimating the number of biconnected components.

\begin{definition}[Bad Vertex]
We say a vertex $v \in V$ is a bad vertex, if one of the following holds:
\begin{itemize}
    \item degree of $v$ is 1,
    \item $v$ is a cut vertex with degree 2.
\end{itemize}
\end{definition}

In the following series of lemmas, we bound the cost of graphic TSP based on the size of maximum path cover and number of bad vertices. Almost all the steps of this part are similar to the algorithm for graphic TSP of \cite{chen2020} --- except the path cover subroutine that we use instead of maximal matching subroutine. We restate some of the useful lemmas to achieve the approximation bound that the black-box algorithm can get, and in the next subsection we improve this bound. First, we prove that if the size of the maximum path cover is small, the cost of graphic TSP is bounded away from $n$.

\begin{claim}\label{lem: lower-bound-path}
If the size of maximum path cover of graph $G$ is at most $\rho$, then the cost of graphic TSP is at least $2n - \rho$.  
\end{claim}
\begin{myproof}
Let $(v_0, v_2, \ldots, v_n = v_0)$ be the optimal graphic TSP of graph $G$. Note that the subgraph induced by weight-one edges of this cycle is a solution for path cover. Hence, at most $\rho$ edges in cycle $(v_0, v_2, \ldots, v_n = v_0)$ have weight one. All the remaining edges have a weight of at least two which implies the claimed bound.
\end{myproof}

Furthermore, the following lemma from \cite{chen2020}, provides a lower bound for a graphic TSP of graph in terms of number of bad vertices.

\begin{lemma}[{\cite[Lemma~2.8]{chen2020}}]\label{lem: lower-bound-bad-vertex}
If the number of bad vertices of graph $G$ is at least $\beta$, then the cost of graphic TSP is at least $n + \beta - 2$.
\end{lemma}

Chen et al. \cite{chen2020} showed that in a biconnected graph, if there exists a matching of large size, the cost of graphic TSP is significantly smaller than $2n$. 

\begin{lemma}[{\cite[Lemma~2.7]{chen2020}}]\label{lem: tsp-matching-bound}
Let $G$ be a graph and $M$ be a matching of $G$. Then the cost of graphic TSP is at most $2n - |M|$.
\end{lemma}

\begin{lemma}[{\cite[Lemma~2.11]{chen2020}}]\label{lem: non-bridge-matching}
Let $G$ be a graph and $M'$ be a matching that none of its edges is a bridge. Then the cost of graphic TSP is at most $2n - \frac{2}{3}|M'|$.
\end{lemma}

We now upper bound the cost of graphic TSP in terms of size of maximum path cover and number of bad vertices.

\begin{lemma}\label{lem:bound-path-cover-bad-vertices-tsp}
If the size of maximum path cover of graph $G$ is $\rho(G)$ and it has $\beta$ bad vertices, then the cost of graphic TSP is at most $2n - \frac{1}{5}(\rho(G) - 2\beta)$.
\end{lemma}
\begin{myproof}
Let $l$ be the number of non-trivial biconnected components and $M'$ be a maximum matching in graph $G$ that none of its edges is a bridge. Also, let $B$ be the number of bridges in $G$. By the proof of Lemma 2.9 of \cite{chen2020}, the cost of graphic TSP is at most $\min\{2n - \frac{2}{3}|M'|, 2n - l\}$. Note that there are at least $\rho(G) - B$ edges of the maximum path cover that are not bridge. Since all non-bridge edges of the maximum path cover are still union of several disjoint paths, there exists a matching with size of at least half of the edges of these paths. Hence, there exist a matching of size at least $\frac{1}{2}(\rho(G) - B)$. On the other hand, in the proof of the same lemma, they showed that $l \geq \frac{B}{2} - \beta$ which implies that the cost of graphic TSP is at most
\begin{align*}
\min\left\{2n - \frac{2}{3}|M'|, 2n - l\right\} \leq \min\left\{2n - \frac{1}{3}(\rho(G) - B), 2n - \frac{B}{2} + \beta\right\}.
\end{align*}
There are two possible cases:
\begin{itemize}
    \item If $B \leq \frac{2}{5}\rho(G) + \frac{6}{5}\beta$, then we have 
    \begin{align*}
        2n - \frac{1}{3}(\rho(G) - B) \leq 2n - \frac{1}{3}(\rho(G) - \frac{2}{5}\rho(G) - \frac{6}{5}\beta) = 2n - \frac{1}{5}(\rho(G) - 2\beta).
    \end{align*}
    \item If $B > \frac{2}{5}\rho(G) + \frac{6}{5}\beta$, then we have
    \begin{align*}
        2n - \frac{B}{2} + \beta \leq 2n - \frac{1}{5}\rho(G) - \frac{3}{5}\beta + \beta = 2n - \frac{1}{5}(\rho(G) - 2\beta).
    \end{align*}
\end{itemize}
Therefore, the cost of graphic TSP is at most
\begin{align*}
    \min\left\{2n - \frac{1}{3}(\rho(G) - B), 2n - \frac{B}{2} + \beta\right\} \leq 2n - \frac{1}{5}(\rho(G) - 2\beta). \qquad \qedhere
\end{align*}
\end{myproof}

Now we are ready to introduce the first algorithm for estimating the cost of graphic TSP, which uses our maximum path cover subroutine instead of the matching subroutine as a black-box. In \Cref{alg: graphic-tsp}, we first estimate the size of the maximum path cover and the number of bad vertices of the graph and report the graphic TSP cost in terms of the two estimations. The subroutine used for counting number of bad vertices is similar to the one in section 2.2 of \cite{chen2020}.

\begin{lemma}[\cite{chen2020}]\label{lem: bad-vertices-estimator}
Let $\beta$ be the number of bad vertices. For any constant $\epsilon > 0$, there exists an algorithm that w.h.p estimates the number of bad vertices $\tilde{\beta}$, such that $\beta \leq \tilde{\beta} \leq \beta + \epsilon n $, in $\tilde{O}(n)$ running time.
\end{lemma}

\begin{algorithm}
\caption{First algorithm for graphic TSP.}
\label{alg: graphic-tsp}
    Construct $\hat{G} = (V_{\hat{G}}, E_{\hat{G}})$ implicitly as desribed in \Cref{sec: disjoint-paths-estimate}.
    
    Let $\tilde{\rho}$ be the output of \Cref{alg: disjoint-paths} on $\hat{G}$.
    
    Let $\tilde{\beta}$ be the estimate of number of bad vertices.
    
    $\tilde{\tau} = 2n - \frac{1}{5}(\tilde{\rho} - 2\tilde{\beta})$
    
    \Return $\tilde{\tau}$
\end{algorithm}

\begin{lemma}\label{lem: graphic-tsp-approx1}
Let $\tilde{\tau}$ be the output of \Cref{alg: graphic-tsp} and $\tspsize{V}$ be the cost of graphic TSP of graph $G = (V, E)$. With high probability,
$$
\tspsize{V} \leq \tilde{\tau} \leq (\frac{19}{10} + \frac{1}{K})\cdot \tspsize{V},
$$
where $K$ is the parameter which is defined in \algtwo{}.
\end{lemma}

\begin{myproof}
Let $\beta$ be the number of bad vertices. By \Cref{lem: error-bound} and \Cref{lem: bad-vertices-estimator} $\tilde{\rho} \leq \rho(G)$ and $\beta \leq \tilde{\beta}$. Hence, we have $\tspsize{V} \leq \tilde{\tau}$ by \Cref{lem:bound-path-cover-bad-vertices-tsp}.

By \Cref{lem: error-bound} and \Cref{lem: bad-vertices-estimator}, we can estimate $\rho(G)$ and $\beta$ such that $\left(\frac{1}{2} - \frac{1}{K} \right)\rho(G) - \frac{n}{K} \leq \tilde{\rho}$ and $\tilde{\beta} \leq \beta + \frac{n}{K}$, if we choose $\epsilon = \frac{1}{K}$. Thus, we have
\begin{align*}
    \tilde{\tau} & \leq 2n - \frac{1}{5}\left( (\frac{1}{2} - \frac{1}{K})\cdot \rho(G) - \frac{n}{K} - 2(\beta + \frac{n}{K}) \right) \\
    & \leq 2n - \frac{1}{5}(\frac{\rho(G)}{2} - 2\beta) + \frac{4n}{5K}.& (\text{Since } \rho(G) \leq n). \\
\end{align*}
On the other hand, assume that the approximation ratio that the algorithm obtains is $\alpha + 1/K$ for some $\alpha \leq 2$. Thus, we get 
\begin{align*}
    (\alpha+ \frac{1}{K})\cdot \tau(V) & \geq \alpha\cdot\tau(V) + \frac{n}{K} & (\text{Since } \tau(V) \geq n) \\ 
    & \geq \alpha\cdot \max\{2n - \rho(G), n + \beta - 2\} + \frac{n}{K} & (\text{By  \Cref{lem: lower-bound-path}  and  \Cref{lem: lower-bound-bad-vertex}})\\
    & \geq \alpha\cdot \max\{2n - \rho(G), n + \beta\} + \frac{n}{K} - 4.
\end{align*}
So in order to show that $\tilde{\tau} \leq (\alpha + \frac{1}{K}) \cdot\tspsize{V}$, it is sufficient to show that 
\begin{align*}
2n - \frac{1}{5}(\frac{\rho(G)}{2} - 2\beta) + \frac{4n}{5K} \leq \alpha\cdot \max\{2n - \rho(G), n + \beta\} + \frac{n}{K} - 4.
\end{align*}
If $n$ is large enough, then we have $\frac{4n}{5K} \leq \frac{n}{K} - 4$, which implies that we need to prove
\begin{align*}
    2n - \frac{1}{5}(\frac{\rho(G)}{2} - 2\beta) \leq \alpha\cdot \max\{2n - \rho(G), n + \beta\}.
\end{align*}
Now, let $\rho(G) = xn$ and $\beta = yn$ for $0 \leq x \leq 1$ and $ 0\leq y \leq 1$. To obtain $\alpha$, it suffices to solve the following program
\begin{align*}
\begin{array}{ll@{}ll}
\text{maximize}  & \alpha \\
\vspace{5px}\text{subject to}& \frac{2 - \frac{1}{5}(\frac{x}{2} - 2y)}{\max\{2-x, 1+y\}} \leq \alpha,\\
                & 0\leq x \leq 1,\\
                 & 0\leq y \leq 1.
\end{array}
\end{align*}
This is a constant size program that can be easily solved; the solution is $19/10$.\footnote{See e.g. this WolframAlpha \href{https://www.wolframalpha.com/input?i=Maximize\%5B\%282-1\%2F5*\%28x\%2F2+-+2y\%29\%29+\%2F+Max\%5B2-x\%2C+1\%2By\%5D\%2C+\%7B0\%3C\%3D+x+\%3C\%3D1\%2C+0\%3C\%3Dy\%3C\%3D1\%7D\%2C+\%7Bx\%2Cy\%7D\%5D}{link}.} This completes the proof.
\end{myproof}

\begin{theorem}\label{thm: graphic-tsp}
Let $\tau(V)$ be the cost of graphic TSP of graph $G = (V, E)$. For any $\epsilon > 0$, there exists an algorithm that estimate the cost of graphic TSP, $\tilde{\tau}$, such that
$$
\tau(V) \leq \tilde{\tau} \leq (\frac{19}{10} + \epsilon) \cdot \tau(V),
$$
w.h.p in $\tilde{O}(n)$ running time.
\end{theorem}
\begin{myproof}
Let $\tilde{\tau}$ be the output of \Cref{alg: graphic-tsp}. If we choose $K = \frac{1}{\epsilon}$, then by \Cref{lem: graphic-tsp-approx1}, we have 
$$
\tau(V) \leq \tilde{\tau} \leq (\frac{19}{10} + \epsilon) \cdot \tau(V).
$$
Also, by \Cref{thm: disjoint-paths-algorithm} and \Cref{lem: bad-vertices-estimator}, estimating $\tilde{\rho}$ and $\tilde{\beta}$ can be done in $\tilde{O}(n)$ time.
\end{myproof}

\section{Further Improvement for Graphic TSP}\label{sec:graphic-TSP-major}
In this section, we design an algorithm to estimate the number of bridges in given graph $G$. Equipped with this tool, we are able to estimate the number of non-bridge edges in the path cover which helps to improve the approximation ratio. Before describing the techniques for estimating the number of bridges, we prove the following lemma that provides a lower bound on the cost of graphic TSP based on the number of bridges in the graph.

\begin{claim}\label{lem:brdige-lower-bound}
If the number of bridges in the graph $G$ is at least $B$, then the cost of the graphic TSP is at least $n + B$.
\end{claim}
\begin{myproof}
Since the metric in the graphic TSP is corresponding to the shortest path distances in graph $G$, then a TSP tour is corresponding to a closed walk that contains all vertices. Thus, each bridge should be crossed at least two times in this walk in order for the walk to be closed and cover all vertices. Therefore, the cost of graphic TSP is at least $n + B$.
\end{myproof}

In the following series of lemmas, first, we prove that there are a few bridges that both of their endpoints have a high degree and then we show an efficient way to estimate the number of bridges that have at least one endpoint with a low degree. Combining the above arguments is the main idea to estimate the number of bridges.

\begin{lemma}\label{lem:few-high-deg-bridges}
For any integer $c \geq 2$, there exists at most $\frac{2n}{c}$ bridges that both of their endpoints have a degree larger than $c$. 
\end{lemma}
\begin{myproof}
Let $\mathcal{B}$ be the set of bridges that both of their endpoints have a degree larger than $c$. We construct a tree, $T_\mathcal{B}$, with edge set equal to $\mathcal{B}$ such that each vertex of $T_\mathcal{B}$ corresponds to a component of vertices that are compressed to a single vertex. We construct $T_\mathcal{B}$ iteratively. In the beginning, we consider the bridge-block tree of the original graph. In each step, if there exists a bridge $e = (u,v)$ (note that $u$ and $v$ are vertices of the tree and corresponding to a set of vertices of the original graph) such that at least one of its endpoints has a degree less than or equal to $c$, we merge $u$ with $v$ and add all edges of $u$ to $v$. We continue this process until there is no bridge with an endpoint of degree less than or equal to $c$.

Now, we prove that $|\mathcal{B}| \leq \frac{2n}{c}$. Let $x_v$ denote the number of vertices in the original graph that are compressed to vertex $v \in T_\mathcal{B}$. We remove vertices of $T_\mathcal{B}$ one by one until there is no vertex in the tree. At each step, we remove a leaf $v \in T_\mathcal{B}$ and at the end when only one vertex is remaining, we remove that vertex. Let $y_v$ be the number of incident edges to $v$ in $T_\mathcal{B}$ that are removed before removing $v$. At the time that we are removing leaf $v$, we have $x_v + y_v + 1\geq c + 1$, since the endpoint of the leaf that is the component of $v$ has at most $x_v$ incident edges in the same component in the original graph, $y_v$ incident edges to the other components that are removed before, and there is only one remaining incident edge to other components (the other endpoint of the leaf). Thus, 
\begin{align}\label{eq:x-bound}
    \sum_{v \in T_\mathcal{B}} x_v \geq \sum_{v \in T_\mathcal{B}} (c - y_v) = (|\mathcal{B}| + 1)c - \sum_{v \in T_\mathcal{B}}y_v.
\end{align}

Since vertices of each component are disjoint, we have $\sum_{v \in T_\mathcal{B}} x_v = n$. Moreover, we have $\sum_{v \in T_\mathcal{B}} y_v = |\mathcal{B}|$ since each edge of $\mathcal{B}$ counted when one of its endpoints is deleted from the tree. Combining above bounds and inequality (\ref{eq:x-bound}), we have
\begin{align*}
    n = \sum_{v \in T_\mathcal{B}} x_v \geq  (|\mathcal{B}| + 1)c - |\mathcal{B}|
\end{align*}
Therefore,
\begin{align*}
    |\mathcal{B}| \leq \frac{n - c}{c - 1} \leq \frac{2n}{c},
\end{align*}
where the last inequality holds for sufficiently large $n$.
\end{myproof}

\begin{lemma}\label{lem:test-bridge}
Let $c \geq 2$ be a constant and $u$ is a vertex such that $\deg(u) \leq c$. Then we can test if each of incident edges of $u$ is a bridge in  $O(n)$ total running time.
\end{lemma}
\begin{myproof}
We can query all neighbors of $u$ in $O(n)$. Assume that $\{v_1, v_2, \ldots, v_r\}$ are neighbors of $u$ for $r \leq c$. Now we divide the vertices of the graph except $u$ into $r$ sets $V_1, V_2, \ldots, V_r$. For each vertex $w \neq u$, we query the distance of $w$ to all $\{v_1, v_2, \ldots, v_r\}$. Let $v_i$ be the closest one to $w$ (if there is a tie, choose the one with the lowest index). Then we put $w$ in $V_i$. Note that since $c$ is a constant and $r \leq c$, this step can be done in $O(n)$.

Now we claim that $(u, v_j)$ is a bridge iff the following conditions hold:
\begin{itemize}
    \item For each $w \in V_j$ and $i \neq j$, $d(w, v_i) - d(w, v_j) = 2.$
    \item For each $w \in V_i$ such that $i \neq j$, $d(w, v_j) - d(w, v_i) = 2.$
\end{itemize}

Suppose that $e = (u, v_j)$ is a bridge. Since removing $e$ creates two connected components $C_u$ and $C_{v_j}$, all vertices in $C_{v_j}$ (resp. $C_u$) have a closer distance to $v_j$ (resp. $u$). In other words, all shortest paths between $w \in V_j$ to $v_i$ for $i \neq j$, cross edges $(v_j,u)$ and $(u, v_i)$. In addition, all the shortest paths between $w \in V_i$ and $v_j$ for $i \neq j$, cross edges $(v_j,u)$ and $(u, v_i)$. Therefore, both conditions hold.

Now suppose that $e = (u, v_j)$ is not a bridge. In this case, there must be an edge between $V_j$ and at least one of $V_i$ as otherwise, $V_j$ will be disconnected from the rest of the graph by removing $e$. Without loss of generality, assume that this edge is $(w, w')$ such that $w \in V_j$, $w' \in V_i$, and $i \neq j$. Also, w.l.o.g., we assume $d(w,v_j) \leq d(w', v_i)$. Since there is an edge between $w$ and $w'$, we have $d(w', v_j) \leq 1 + d(w, v_j) \leq 1 + d(w', v_i)$, which contradicts the conditions.

To test whether the conditions hold, we need to query the distance of each vertex to all $\{v_0, v_1, \ldots, v_r\}$ which can be done in $O(n)$ in total since $r$ is a constant.
\end{myproof}

\begin{lemma}\label{lem:bridge-estimator}
Let $B$ be the number of bridges in graph $G$. For any $\epsilon > 0$, there exists an algorithm that outputs an estimate $\tilde{B}$ in $\tilde{O}(n)$ such that $B \leq \tilde{B} \leq B + \epsilon n$.
\end{lemma}
\begin{myproof}
By \Cref{lem:few-high-deg-bridges}, there are at most $\frac{\epsilon n}{2}$ bridges with both endpoints have degree larger than $\frac{4}{\epsilon}$. Let $\hat{B}$ be the number of bridges that at least one of their endpoint has degree of at most $\frac{4}{\epsilon}$. Thus,
\begin{align}\label{eq:b-range}
    B - \frac{\epsilon n}{2} \leq \hat{B} \leq B.
\end{align}

We sample $r = 256\cdot \epsilon^{-4} \cdot \log n$ vertices uniformly at random with replacement. Let $u$ be the $i$-th sampled vertex. If the degree of the vertex is larger than $\frac{4}{\epsilon}$, we let $X_i = 0$. Otherwise, let $\{v_1, v_2, \ldots, v_k\}$ be the neighbors of $u$ where $k \leq \frac{4}{\epsilon}$. By \Cref{lem:test-bridge}, we can test if each of the incident edges of $u$ is a bridge in $O(n)$ total running time. For each edge $(u, v_j)$ if $\deg(u) < \deg(v_j)$ or $\deg(u) = \deg(v_j)$ and index of $u$ is smaller than $v_j$, we test if the edge is a bridge or not. Let $X_i$ show the number of successful tests for incident edges of $u$. Note that in the above algorithm, each bridge with low-degree endpoints only counted once. 

Let $\bar{X} = (\sum_{i}^{r} X_i) / r$ and $n\bar{X} + \frac{3\epsilon n}{4}$ be our final estimate of the number of bridges. Hence, $\E[\bar{X}] = \hat{B}/n$. Since $\bar{X}$ is the average of $r$ independent random variables such that $0 \leq X_i \leq 4/\epsilon$, by Hoeffding's inequality (\Cref{prop:hoeffding}) we obtain
\begin{align*}
    \Pr\left[\left|\bar{X} - \E[\bar{X}]\right| \geq \frac{\epsilon}{4}\right] \leq 2\exp\left(-\frac{r\epsilon^4}{128}\right) = \frac{2}{n^2},
\end{align*}
where the last inequality follows from $r = 256\cdot \epsilon^{-4} \cdot \log n$. Therefore, with probability of $1 - \frac{2}{n^2}$,
\begin{align*}
    n\bar{X} & \in n\E[\bar{X}] \pm \frac{n\epsilon}{4}\\
    & = \hat{B} \pm \frac{n\epsilon}{4} & (\text{Since } \E[\bar{X}] = \hat{B}/n).
\end{align*}
Combining above range and inequality (\ref{eq:b-range}), we get
\begin{align*}
    B\leq n\bar{X} + \frac{3\epsilon n}{4} \leq B + \epsilon n.
\end{align*}
Since the number of sampled vertices is $r = 256\cdot \epsilon^{-4} \cdot \log n$, the total running time is $\tilde{O}(n)$.
\end{myproof}

Now we are ready to introduce the improved algorithm for graphic TSP.

\begin{algorithm}
\caption{Improved algorithm for graphic TSP.}
\label{alg: graphic-tsp2}
    Construct $\hat{G} = (V_{\hat{G}}, E_{\hat{G}})$ implicitly as described in \Cref{sec: disjoint-paths-estimate}.
    
    Let $\tilde{\rho}$ be the output of \Cref{alg: disjoint-paths} on $\hat{G}$.
    
    Let $\tilde{B}$ be the estimate of the number of bridges. 
    
    $\tilde{\tau} = 2n - \frac{1}{3}(\tilde{\rho} - \tilde{B})$
    
    \Return $\tilde{\tau}$
\end{algorithm}

\begin{lemma}\label{lem: graphic-tsp-approx2}
Let $\tilde{\tau}$ be the output of \Cref{alg: graphic-tsp2} and $\tspsize{V}$ be the cost of graphic TSP of graph $G = (V, E)$. With high probability,
$$
\tspsize{V} \leq \tilde{\tau} \leq (\frac{11}{6} + \frac{1}{K}) \cdot \tspsize{V},
$$
where $K$ is the parameter which is defined in \algtwo{}.
\end{lemma}
\begin{myproof}
Let $\rho(G)$ be the size of maximum path cover and $B$ be the number of bridges in the graph. There are at least $\rho(G) - B$ edges of maximum path cover that are not bridge. These edges construct disjoint paths which implies there exists a matching of size $\frac{1}{2}(\rho(G) - B)$ that none of its edges is a bridge. Hence, by \Cref{lem: non-bridge-matching}, the cost of graphic TSP is at most $2n - \frac{1}{3}(\rho(G) - B)$. Therefore, since $\tilde{\rho} \leq \rho(G)$ and $B \leq \tilde{B}$, we get $\tspsize{V} \leq \tilde{\tau}$.

By \Cref{lem: error-bound} and \Cref{lem:bridge-estimator}, we have $\left(\frac{1}{2} - \frac{1}{K} \right)\cdot\rho(G) - \frac{n}{K}\leq \tilde{\rho}$ and $\tilde{B} \leq B + \frac{n}{K}$ which implies
\begin{align*}
    \tilde{\tau} & \leq 2n - \frac{1}{3}\left( (\frac{1}{2} - \frac{1}{K})\cdot\rho(G) - (B + \frac{n}{K}) \right) \\
    & \leq 2n - \frac{1}{3}(\frac{\rho(G)}{2} - B) + \frac{2n}{K} & (\text{Since } \rho(G) \leq n).
\end{align*}
Also, assume that the approximation ratio that the algorithm obtains is $\alpha + 2/K$ for some $\alpha \leq 2$. Thus,
\begin{align*}
    (\alpha + \frac{2}{K}) \cdot \tspsize{V} & \geq \alpha \cdot \tspsize{V} + \frac{2n}{K} & (\text{Since } \tspsize{V} \geq n) \\ 
    & \geq \alpha\cdot \max\{2n - \rho(G), n + B\} + \frac{2n}{K} & (\text{By \Cref{lem: lower-bound-path}  and  \Cref{lem:brdige-lower-bound}}).
\end{align*}
Therefore, to show that $(\alpha + \frac{1}{K})\cdot\tspsize{V} \geq \tilde{\tau}$, it is sufficient to show 
\begin{align*}
    2n - \frac{1}{3}(\frac{\rho(G)}{2} - B) \leq \alpha\cdot \max\{2n - \rho(G), n + B\}.
\end{align*}

Now, let $\rho(G) = xn$ and $B = yn$ for $0 \leq x \leq 1$ and $ 0\leq y \leq 1$. To obtain $\alpha$, we write the following maximization problem,
\begin{align*}
\begin{array}{ll@{}ll}
\text{maximize}  & \alpha \\
\vspace{5px}\text{subject to}& \frac{2 - \frac{1}{3}(\frac{x}{2} - y)}{\max\{2-x, 1+y\}} \leq \alpha,\\
                & 0\leq x \leq 1,\\
                 & 0\leq y \leq 1.
\end{array}
\end{align*}
The solution to this problem is $11/6$.\footnote{See e.g. this WolframAlpha  \href{https://www.wolframalpha.com/input?i=Maximize\%5B\%282-1\%2F3*\%28x\%2F2+-+y\%29\%29+\%2F+Max\%5B2-x\%2C+1\%2By\%5D\%2C+\%7B0\%3C\%3D+x+\%3C\%3D1\%2C+0\%3C\%3Dy\%3C\%3D1\%7D\%2C+\%7Bx\%2Cy\%7D\%5D}{link}.} This completes the proof.
\end{myproof}

\begin{theorem}\label{thm: graphic-tsp2}
Let $\tspsize{V}$ be the cost of graphic TSP of graph $G = (V, E)$. For any $\epsilon > 0$, there exists an algorithm that estimate the cost of graphic TSP, $\tilde{\tau}$, such that
$$
\tspsize{V} \leq \tilde{\tau} \leq (\frac{11}{6} + \epsilon) \cdot \tspsize{V},
$$
w.h.p in $\tilde{O}(n)$ running time. \vspace{3px}
\end{theorem}
\begin{myproof}
Let $\tilde{\tau}$ be the output of \Cref{alg: graphic-tsp2}. If we choose $K = \frac{1}{\epsilon}$, then by \Cref{lem: graphic-tsp-approx2}, we have 
$$
\tspsize{V} \leq \tilde{\tau} \leq (\frac{11}{6} + \epsilon) \cdot \tspsize{V}.
$$
Also, by \Cref{thm: disjoint-paths-algorithm} and \Cref{lem:bridge-estimator}, estimating $\tilde{\rho}$ and $\tilde{B}$ can be done in $\tilde{O}(n)$ time.
\end{myproof}

\section{A Slightly Subquadratic Algorithm for Graphic TSP}\label{sec:graphic-TSP-slightly-sub}

With the recent advances in designing sublinear algorithms for maximum matching \cite{ behnezhadroghanirubinstein2022, BehnezhadRRS-ArXiv22, sayanSublinear2022, bhattacharya2023dynamic}, we can now achieve a more precise estimation of the size of graphic TSP at the cost of increased running time. In this section, we present an algorithm that approximates the graphic TSP with an accuracy within a factor of $5/3 + \epsilon$ in $O(n^{2-\Omega_\epsilon(1)})$. We use the following result by Bhattacharya, Kiss, and Saranurak \cite{bhattacharya2023dynamic} to design our slightly subquadratic graphic TSP estimator.

\begin{proposition}\label{prop:matching-estimator}
    Suppose that we have access to the adjacency matrix of graph $G$. Then, There exists an algorithm that estimates the size of the maximum matching of graph $G$, $\tilde \mu$, such that
    \begin{align*}
        \mu(G) - \epsilon n \leq \tilde \mu \leq \mu(G),
    \end{align*}
    w.h.p in $n^{2-\Omega_\epsilon(1)}$ time.
\end{proposition}

Combining this algorithm with the framework of Chen, Kannan, and Khanna \cite{chen2020} implies a $(13/7+\epsilon)$-approximation for graphic TSP in $n^{2-\Omega_\epsilon(1)}$ time. Our algorithm in \Cref{sec:graphic-TSP-major} makes an improvement on both the running time and approximation ratio for the graphic TSP over this recent result. Furthermore, using \Cref{prop:matching-estimator} and our estimator for counting the number of bridges, we are able to obtain a 5/3-approximation ratio. Now we are ready to propose our algorithm.

\begin{algorithm}
\caption{Slightly subquadratic algorithm for graphic TSP.}
\label{alg: graphic-tsp3}
    
    Let $\tilde{\mu}$ be the output of \Cref{prop:matching-estimator} on $G$.
    
    Let $\tilde{B}$ be the estimate of the number of bridges. 
    
    $\tilde{\tau} = 2n - \frac{1}{3}(\tilde{\mu} - \tilde{B})$
    
    \Return $\tilde{\tau}$
\end{algorithm}

\begin{lemma}\label{lem: graphic-tsp-approx3}
Let $\tilde{\tau}$ be the output of \Cref{alg: graphic-tsp3} and $\tspsize{V}$ be the cost of graphic TSP of graph $G = (V, E)$. With high probability,
$$
\tspsize{V} \leq \tilde{\tau} \leq (\frac{5}{3} + \epsilon) \cdot \tspsize{V}.
$$
\end{lemma}
\begin{proof}
    Let $B$ be the number of bridges in $G$. There exists a matching with a size of at least $\mu(G) - B$ that none of its edges is a bridge. Thus, by \Cref{lem: non-bridge-matching}, it holds $\tau(V) \leq 2n - \frac{1}{3}(\mu(G) - B)$. Combining with $\tilde\mu \leq \mu(G)$ and $B \leq \tilde B$, we get $\tau(V) \leq \tilde\tau$.

    Additionally, by \Cref{prop:matching-estimator} and \Cref{lem:bridge-estimator}, we have $\tilde \mu \geq \mu(G) - \epsilon n$ and $\tilde B \leq B + \epsilon n$. Thus,
    \begin{align*}
    \tilde{\tau} & \leq 2n - \frac{1}{3}\left( (\mu(G) - \epsilon n) - (B + \epsilon n) \right) \\
    & \leq 2n - \frac{1}{3}(\mu(G) - B) + 2\epsilon n & (\text{Since } \mu(G) \leq n).
    \end{align*}
    Assume the approximation ratio of the algorithm is $\alpha + 2\epsilon$ for some $\alpha \leq 2$, we must have
    \begin{align*}
    (\alpha + 2\epsilon) \cdot \tspsize{V} & \geq \alpha \cdot \tspsize{V} + 2\epsilon n & (\text{Since } \tspsize{V} \geq n) \\ 
    & \geq \alpha\cdot \max\{2n - \mu(G), n + B\} + 2\epsilon n. & (\text{By \Cref{lem: tsp-matching-bound}  and  \Cref{lem:brdige-lower-bound}})
\end{align*}
In order to show $(\alpha + 2\epsilon)\cdot\tspsize{V} \geq \tilde{\tau}$, it is sufficient to show 
\begin{align*}
    2n - \frac{1}{3}(\mu(G) - B) \leq \alpha\cdot \max\{2n - \mu(G), n + B\}.
\end{align*}

Let $\mu(G) = xn$ and $B = yn$ for $0 \leq x \leq 1$ and $ 0\leq y \leq 1$. To obtain $\alpha$, we write the following maximization problem,
\begin{align*}
\begin{array}{ll@{}ll}
\text{maximize}  & \alpha \\
\vspace{5px}\text{subject to}& \frac{2 - \frac{1}{3}(x - y)}{\max\{2-x, 1+y\}} \leq \alpha,\\
                & 0\leq x \leq 1,\\
                 & 0\leq y \leq 1.
\end{array}
\end{align*}
The solution to this problem is $5/3$.\footnote{See e.g. this WolframAlpha  \href{https://www.wolframalpha.com/input?i=Maximize\%5B\%282-1\%2F3*\%28x+-+y\%29\%29+\%2F+Max\%5B2-x\%2C+1\%2By\%5D\%2C+\%7B0\%3C\%3D+x+\%3C\%3D1\%2C+0\%3C\%3Dy\%3C\%3D1\%7D\%2C+\%7Bx\%2Cy\%7D\%5D}{link}.}
\end{proof}

\begin{theorem}\label{thm: graphic-tsp3}
Let $\tspsize{V}$ be the cost of graphic TSP of graph $G = (V, E)$. For any $\epsilon > 0$, there exists an algorithm that estimate the cost of graphic TSP, $\tilde{\tau}$, such that
$$
\tspsize{V} \leq \tilde{\tau} \leq (\frac{5}{3} + \epsilon) \cdot \tspsize{V},
$$
w.h.p in $n^{2-\Omega_\epsilon(1)}$ running time. \vspace{3px}
\end{theorem}
\begin{myproof}
Let $\tilde{\tau}$ be the output of \Cref{alg: graphic-tsp3}. By \Cref{lem: graphic-tsp-approx2}, we have 
$$
\tspsize{V} \leq \tilde{\tau} \leq (\frac{5}{3} + \epsilon) \cdot \tspsize{V}.
$$
Moreover, by \Cref{prop:matching-estimator} and \Cref{lem:bridge-estimator}, estimating $\tilde{\mu}$ and $\tilde{B}$ can be done in $n^{2-\Omega_\epsilon(1)}$ time.
\end{myproof}

\section{Lower Bound for Approximating Maximum Path Cover}\label{sec:lowerbound}

\subsection{``Conditional'' Hardness for the Approximation Ratio}
In this section, we prove that if there exists a constant $\alpha > 0$ and an algorithm that returns a $(\frac{1}{2} + \alpha)$-approximate estimate for the size of maximum path cover in $\tilde{O}(n)$ time in a bipartite graph, then there is a $(\frac{1}{2} + \alpha)$-approximate algorithm for estimating the maximum matching size in $\tilde{O}(n)$ time. This remains an important open problem in the study of sublinear time maximum matching algorithms. See in particular \cite{BehnezhadRRS-ArXiv22}. This implies that short of a major result in the study maximum matchings in the sublinear time model, which have received significant attention in the literature (see \cite{YoshidaYI09,behnezhad2021,BehnezhadRRS-ArXiv22, behnezhadroghanirubinstein2022, sayanSublinear2022} and references therein), our path cover algorithm has an optimal approximation ratio.

Let $G = (V, U, E)$ be a bipartite graph.  We construct a graph $G' = (V', U', E')$ such that a better than $\frac{1}{2}$-approximate estimate of maximum path cover on $G'$ leads to a better than $\frac{1}{2}$-approximate estimate of maximum matching in $G$. Let $r$ be a large constant. We create $r$ copies of $G$, showing the $i$-th copy with $G_i=(V_i, U_i, E)$. Also, we create another $r-1$ copies $H_{1}, \ldots, H_{r-1}$ of $G$ with $H_i = (V_i, U_{i+1}, E)$. Now we let the $G' = (\bigcup_{i=1}^r G_i) \cup (\bigcup_{i=1}^{r-1}H_i)$. Now we claim that the size of maximum path cover of the graph $G'$ is roughly $2r\cdot \mu(G)$ which can be used as an estimator for the maximum matching of $G$.

\begin{figure}[htbp]
\begin{center}
  \includegraphics[scale=0.55]{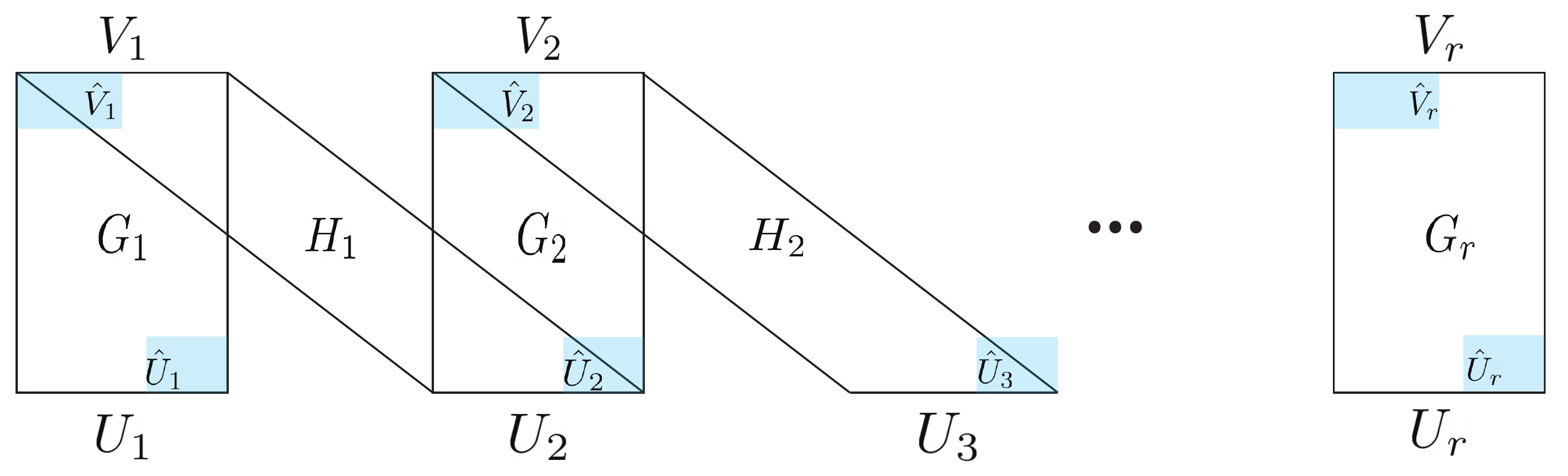}
  \caption{Illustration of graph $G' = (V', U', E')$. Each $G_i$ is shown by a rectangle and each $H_i$ is shown by a parallelogram. Top and bottom horizontal lines illustrate $V_i$ and $U_i$. Blue highlighted parts represent the vertex cover of the graph.}\label{fig:hardness}
  \end{center}
\end{figure}

Before proving the main result of this section, we characterize some properties of the constructed graph $G'$.

\begin{claim}\label{clm:matching-new-graph}
$\mu(G') = r\cdot \mu(G)$.
\end{claim}
\begin{myproof}
First, since all graphs $\{G_i\}_{i=1}^r$ are the same as $G$ and are vertex-disjoint, if we consider the maximum matching of $G$ in each of the $r$ graphs, we will have a matching of size $r\cdot \mu(G)$. Thus, $\mu(G') \geq r\cdot \mu(G)$.

Let $\hat{V} \cup \hat{U}$ be the minimum vertex cover of $G$ such that $\hat{V} \in V$ and $\hat{U} \in U$. By K\"{o}nig's Theorem (\Cref{prop:konig}), we have $|\hat{V} \cup \hat{U}| = \mu(G)$. Now we show there exists a vertex cover of size $r\cdot \mu(G)$ for graph $G'$. Let $\hat{V}_i \in V_i$ (resp. $\hat{U}_i \in U_i$) be the copy of vertices $V$ (resp. $U$) in graph $G_i$. We claim $(\bigcup_{i=1}^r \hat{V}_i \cup \hat{U}_i)$ is a vertex cover for $G'$. If an edge is in $G_i$, then at least one of its endpoints is in $\hat{V}_i \cup \hat{U}_i$ since $\hat{V}_i \cup \hat{U}_i$ is a vertex cover of $G_i$. Moreover, by the construction, $\hat{V}_i \cup \hat{U}_{i+1}$ is a vertex cover of $H_i$. Hence, each edge of $H_i$ is also covered by the vertex cover. Therefore, since there exists a vertex cover of size $|(\bigcup_{i=1}^r \hat{V}_i \cup \hat{U}_i)| = r\cdot |\hat{V} \cup \hat{U}| = r\cdot \mu(G)$, then we have $\mu(G') \leq r\cdot \mu(G)$ which completes the proof.
\end{myproof}

\begin{observation}\label{obs:path-cover-to-matching}
It holds $(2r-1)\cdot\mu(G) \leq \rho(G') \leq 2r\cdot \mu(G)$.
\end{observation}
\begin{myproof}
Since the union of maximum matching of all graphs $\{G_i\}_{i=1}^r$ and $\{H_i\}_{i=1}^{r-1}$ creates a path cover, we get $(2r-1)\cdot\mu(G) \leq \rho(G')$. Futhermore, if there exists a path cover of size larger than $2r\cdot \mu(G)$, then the maximum matching of these paths will be larger than $r\cdot \mu(G)$ which contradicts \Cref{clm:matching-new-graph}. Thus, $\rho(G') \leq 2r\cdot \mu(G)$.
\end{myproof}

Now we are ready to show the reduction.

\begin{lemma}
For any constant $\alpha > 0$, if there exists an algorithm that can estimate the maximum path cover within a $(\frac{1}{2} + \alpha)$-factor in $O(T(n))$ time, then the same algorithm can be used to estimate the maximum matching of bipartite graph $G$ within a $(1 - \epsilon)\cdot (\frac{1}{2} + \alpha)$-factor in $O(T(n/\epsilon))$ time.
\end{lemma}
\begin{myproof}
We construct graph $G'$ as described at the beginning of the section with $r = \frac{1}{2\epsilon}$. By \Cref{obs:path-cover-to-matching}, $(\frac{1}{\epsilon}-1)\cdot\mu(G) \leq \rho(G') \leq \frac{1}{\epsilon}\cdot \mu(G)$. Let $\tilde{\rho}$ be the estimate of the algorithm for the maximum path cover of $G'$. Hence, we have 
\begin{align*}
    (\frac{1}{2} + \alpha)(\frac{1}{\epsilon}-1)\cdot\mu(G) \leq \tilde{\rho} \leq \frac{1}{\epsilon}\cdot \mu(G).
\end{align*}
Now let $\tilde{\mu} = \epsilon\cdot\tilde{\rho}$ be the estimate for the maximum matching of $G$. Hence, 
\begin{align*}
    (1-\epsilon)(\frac{1}{2} + \alpha)\cdot\mu(G) \leq \tilde{\mu} \leq \mu(G).
\end{align*}
Since the number of vertices and number of edges of $G'$ is $r=\frac{1}{2\epsilon}$ times more than $G$, then the running time will be $O(T(n/\epsilon))$.
\end{myproof}

A reduction to matchings can also be proved for $(1,2)$-TSP, albeit with an extra promise for the matching instance that the matching is either perfect or half-perfect. This problem, formalized below, also remains open in the study matchings. We show that a better than $1.5$-approximation for $(1,2)$-TSP in $\widetilde{O}(n)$ time would resolve this question.

\begin{problem}\label{problem:matching-promised}
Suppose that we are given a bipartite graph $G=(L, R, E)$ with $|L| = |R| = n$ and are promised that either $\mu(G) = n$ or $\mu(G) = (\frac{1}{2}+\epsilon)n/2$ for any desirably small constant $\epsilon > 0$. Provided adjacency matrix access to the graph, does there exist an $n^{1+o(1)}$ time algorithm that distinguishes the two?
\end{problem}

\begin{theorem}\label{cor:matching-via-12tsp}
    If there is an algorithm that estimates the size of $(1,2)$-TSP within a $(\frac{3}{2} - \epsilon_0)$-factor for some fixed constant $\epsilon_0 \in (0, \frac{1}{4}]$ in $n^{1+o(1)}$, then \cref{problem:matching-promised} can indeed be solved in $n^{1+o(1)}$ time.
\end{theorem}
\begin{myproof}
Let $G_1$ and $G_2$ be two graphs with $n$ vertices such that $\mu(G_1) = n$ and $\mu(G_2) = (\frac{1}{2} + \frac{\epsilon_0}{16})$. We construct graph $G_1' = (V_1', E_1')$ and $G_2' = (V_2', E_2')$ as described at the beginning of the section with $r = \frac{1}{\epsilon_0}$. By \Cref{obs:path-cover-to-matching}, we have $\rho(G_1') \geq (\frac{2}{\epsilon_0} - 1)n$ and $\rho(G_2') \leq (\frac{1}{\epsilon_0}  + \frac{1}{8})n$. Thus, by \Cref{obs: bounding-TSP-with-DP}, we get
\begin{align*}
    \tspsize{V_1'} \leq \frac{4}{\epsilon_0}n - (\frac{2}{\epsilon_0} - 1)n = (\frac{2}{\epsilon_0} + 1)n,\\
    \tspsize{V_2'} \geq \frac{4}{\epsilon_0}n - (\frac{1}{\epsilon_0} + \frac{1}{8})n - 1 \geq (\frac{3}{\epsilon_0} - \frac{1}{4})n,
\end{align*}
for sufficiently large $n$, which implies
\begin{align*}
    \frac{\tspsize{V_2'}}{\tspsize{V_1}} = \frac{3 - \epsilon_0/4}{2 + \epsilon_0} \geq \frac{3}{2} - \epsilon_0.
\end{align*}
Therefore, the algorithm for (1,2)-TSP can distinguish between $G_1'$ and $G_2'$ which implies \cref{problem:matching-promised} can be solved in $n^{1+o(1)}$ time for $\epsilon = \epsilon_0/16$.
\end{myproof}

\paragraph{Lower Bound for Graphic TSP:} note that we can reduce an instance of (1,2)-TSP to graphic TSP by adding a new vertex and connecting the newly added vertex to all vertices of the graph. Therefore, the $n^{1+o(1)}$ time lower bound also holds for graphic TSP.

\subsection{Information-Theoretic Lower Bounds on the Running Time}\label{sec:hardness}
Since any constant approximation algorithm for estimating maximum path cover can be used to estimate the size of matching within a constant factor, then all of the lower bounds for $O(1)$-approximating maximum matching in sublinear time also hold for $(1)$-approximating maximum path cover in sublinear time. We restate some of these lower bounds along with a short proof (see \cite{behnezhad2021} for a detailed discussion).

\begin{lemma}\label{lem:hardness-adjacency-list}
Any algorithm that estimates maximum path cover within a constant multiplicative factor requires $\Omega(n)$ queries in the adjacency list model.
\end{lemma}
\begin{myproof}
Consider two graphs that the first one does not have any edge and the second one has only a single edge. In order to give any multiplicative approximation for maximum path cover, the algorithm needs to find the edge which requires $\Omega(n)$ queries in the adjacency list model. 
\end{myproof}

\begin{lemma}
Any algorithm that estimates maximum path cover within a constant multiplicative factor require $\Omega(n^2)$ queries in the adjacency matrix model.
\end{lemma}
\begin{myproof}
Consider the same construction as \Cref{lem:hardness-adjacency-list}. To give any multiplicative approximation for maximum path cover, the algorithm needs to find the edge which requires $\Omega(n^2)$ queries in the adjacency matrix model. 
\end{myproof}

\begin{lemma}
Any algorithm that estimates maximum path cover within a multiplicative-additive  requires $\Omega(n)$ queries in the adjacency matrix model.
\end{lemma}
\begin{myproof}
Consider a graph with no edge and a graph with one Hamiltonian cycle and no other edges. In order for the algorithm to distinguish between these two graphs, it must find at least one edge of the second graph which requires $\Omega(n)$ queries in the adjacency matrix model. 
\end{myproof}

There is also a lower bound for multiplicative-additive estimation of matching in adjacency list model \cite{PARNAS2007183} that also holds for maximum path cover.

\begin{lemma}
Any algorithm that estimates maximum path cover within a constant multiplicative-additive factor requires $\Omega(\bar{d})$ queries in the adjacency list model.
\end{lemma}

\paragraph{Acknowledgements.} Mohammad Roghani and Amin Saberi were supported by NSF award 1812919 and ONR award 141912550. Soheil Behnezhad and Aviad Rubinstein were supported by NSF CCF-1954927, and a David and Lucile Packard Fellowship. Soheil Behnezhad was additionally supported by NSF Awards 1942123, 1812919 and by Moses Charikar's Simons Investigator Award.


\bibliographystyle{plainnat}
\bibliography{references}

\appendix

\section{Implementation Details}\label{sec:implementation}

In this section, we discuss why \cref{lem: implementation}, restated below, holds.

\oracleImplementation*

The proof of \cref{lem: implementation} uses standard ideas from the literature \cite{OnakRRR12,behnezhad2021}. The only modification, essentially, is to show that these algorithms also work for multi-graphs. Let us focus on the specific algorithm proposed in \cite[Appendix~A]{behnezhad2021}. Given the adjacency list of a graph $G=(V, E)$, it defines gives a procedure $\lowest(v,i)$ that first draws a random rank $E \to [0, 1]$ on each edge (implicitly), then for any input vertex $v$ and an integer $i \leq \deg_G(v)$, returns a vertex $u$ such that $(v, u)$ is the $i$-th lowest rank edge incident to $v$. It is proved in \cite{behnezhad2021} that if the procedure is called for a fix vertex $v$ and all indices $i$ with $1 \leq i \leq r$, then the total running time is $\tilde{O}(r)$. The only difference between the implementation of our algorithm and the one in \cite{behnezhad2021} is that we have multiple copies of a single edge in the original graph. First, we observe that the procedure $\lowest(v,i)$, in addition to returning the neighbor $u$, can also return the rank of the edge $(v, u)$. (This is explicitly computed by \lowest(v, i) in \cite{behnezhad2021}.) Now let $G'$ be the multigraph with $K$ copies of each edge of $G$. Instead of a multigraph, we can assume that we have $K$ copies of same graph $G$ called $G_1, G_2, \ldots, G_K$. Also, for each $i$, let $\lowest_{G_i}$ be the \lowest{} procedure corresponding to graph $G_i$. For each vertex $v$, we use a balanced binary search tree (BST) that stores the ranks of the lowest incident edge to $v$ in each graph. So at any point during the course of the algorithm, there are at most $K$ different values in the BST of vertex $v$. Now for the next \lowest{} query to the multigraph graph $G'$ for vertex $v$, we can return the minimum edge in the BST of vertex $v$. Since $K$ is a constant and the any query to a BST is answered in $O(\log n)$ time, the total running time will be the same as \cite[Appendix~A]{behnezhad2021} within a $O(\log n)$-factor.

\end{document}